\UseRawInputEncoding 
\documentclass[12pt]{article}
\usepackage[cmtip,arrow]{xy}
\usepackage{pb-diagram,pb-xy}
\usepackage{amsmath}
\usepackage[psamsfonts]{amssymb}
\usepackage{cmmib57}
\usepackage{amsmath,amssymb,amscd}
\usepackage{tikz-cd}

\def\QED{\hskip0.1em\hfill\null\ \null\nobreak\hfill\kern3pt\vbox{\hrule\hbox
   {\vrule\kern1pt\vbox{\kern1.7pt\hbox{$\scriptscriptstyle{QED}$}
    \kern0.2pt}\kern1pt\vrule}\hrule}}

\def\END{\hskip0.1em\hfill\null\ \null\nobreak\hfill\kern3pt\vbox{\hrule\hbox
   {\vrule\kern1pt\vbox{\kern1.7pt\hbox{$\,\,\,\vspace{5pt}$}
    \kern0.2pt}\kern1pt\vrule}\hrule}}
\usepackage{amsmath,amssymb,amsthm,mathrsfs,amscd}
\newtheorem{theorem}{Theorem}[section]
\newtheorem{lemma}{Lemma}[section]
\newtheorem{prop}{Proposition}[section]

\newtheorem{defn}{Definition}[section]
\newtheorem{definition}{Definition}[section]
\newtheorem{example}{Example}[section]

\newtheorem{remark}{Remark}[section]

\DeclareMathOperator{\byd}{{\raisebox{.1ex}{:}{=}}}
\newcommand{\bCd}{\beq \begin{CD}}
\newcommand{\eCd}{\end{CD}\eEq}
\newcommand{\bcd}{\beq \begin{CD}}
\newcommand{\ecd}{\end{CD}\eeq}
\newcommand{\ben}{\begin{enumerate}}
\newcommand{\een}{\end{enumerate}}
\newcommand{\bEq}{\begin{eqnarray}}
\newcommand{\eEq}{\end{eqnarray}}
\newcommand{\beq}{\begin{eqnarray*}}
\newcommand{\eeq}{\end{eqnarray*}}
\newcommand{\bDf}{\begin{definition}\em}
\newcommand{\eDf}{\end{definition}}
\newcommand{\bLm}{\begin{lemma}}
\newcommand{\eLm}{\end{lemma}}
\newcommand{\bPr}{\begin{proposition}}
\newcommand{\ePr}{\end{proposition}}
\newcommand{\bTh}{\begin{theorem}}
\newcommand{\eTh}{\end{theorem}}
\newcommand{\bCr}{\begin{corollary}}
\newcommand{\eCr}{\end{corollary}}
\newcommand{\bRm}{\begin{remark}\em}
\newcommand{\eRm}{\end{remark}}
\newcommand{\bEx}{\begin{example}\em}
\newcommand{\eEx}{\end{example}}

\newcommand{\ie}{{\em i.e$.$} }
\newcommand{\eg}{{\em e.g$.$} }

\newcommand{\der}{\partial}

\newcommand{\cR}{\mathcal{R}}

\newcommand{\wed}{\wedge}

\newcommand{\tht}{\theta}

\newcommand{\lam}{\lambda}
\newcommand{\sig}{\sigma}

\newcommand{\ome}{\omega}

\title{\large{Geometric integration by parts and Lepage equivalents\footnote{Our colleague  and   friend Olga Rossi  passed away in October 2019. This paper 
is an outcome of our collaboration, which we miss heartily, and we dedicate it to her memory.}}}

\author{
{\normalsize 
Marcella Palese$^{a}$, Olga Rossi and  Fabrizio Zanello$^{b}$}
\\ 
{\footnotesize  $^{a}$Department of Mathematics,
University of Torino}
\\
{\footnotesize via C. Alberto 10, 10123 Torino, Italy} 
\\  {\footnotesize e--mail: 
{\sc marcella.palese@unito.it}} 
\\ {\footnotesize $^{b}$Institute of Mathematics, University of G\"ottingen}
\\
{\footnotesize Bunsenstra\ss e 3-5, D-37073  G\"ottingen, Germany} 
\\  {\footnotesize e--mail: 
{\sc fabrizio.zanello@mathematik.uni-goettingen.de}}
}

\date{}
\begin{document}

\maketitle

\begin{abstract}
We compare the integration by parts of contact forms - 
leading to the definition of the interior Euler operator - with the so-called canonical splittings of variational morphisms. 
In particular,  we discuss the possibility of a generalization of the first method to contact forms of lower degree. 
We define a suitable Residual operator for this case and, working out an original
conjecture by Olga Rossi,
we recover the Krupka-Betounes equivalent for first order field theories. A generalization to the second order case is discussed.
\end{abstract}

\noindent {\bf Key words}: Interior Euler operator; Residual operator; geometric integration by parts; Poincar\'e-Cartan form; Lepage equivalent.

\noindent {\bf 2010 MSC}: 53Z05,58A20,58Z05.

\section{Introduction}

The Euler--Lagrange operator can be geometrically described by means of two interrelated geometric objects (and corresponding geometric integration by parts procedures), the one based on the concept of differential forms and exterior differential modulo contact structures, the other based on the interpretation of variational objects as fibered morphisms \cite{HK:HigherOrder,KMS:Natural,FF:Natural} etc.

Following an approach inaugurated by the works of Cartan and Lepage, the finite order variational sequence was introduced and developed by Krupka; see \eg  \cite{Krupka:Variational}. The problem of the representation of the finite order variational sequence (whose objects, we recall, are equivalence classes of local differential forms) has been discussed in terms of the so called {\em interior Euler operator}; see \eg  \cite{KrbekMusi:Representation} and \cite{Palese:VariationalSeq}.

On the other hand variational morphisms \cite{FF:Natural} 
not only provide a geometric formulation of the calculus of variations, but in general of a wide class of differential operators. Their most relevant property is that,
under certain conditions,
they admit canonical and algorithmic splittings and, by the introduction of a connection on the base manifold and a connection on the considered fiber bundle, globality and uniqueness properties of these splittings can be assured.

The aim of this paper is to investigate the relation between these two approaches which use different geometric integration by parts techniques. 

We perform the identification of contact forms with variational morphisms  for $1$-contact forms of degree at most $n+1$, where $n$ is the dimension of the base manifold of the considered fiber bundles. The two integration by parts techniques are directly compared in the case of $1$-contact $n$-horizontal $(n+1)$-forms and
we show that the two approaches are straightly equivalent, whilst in the case of $1$-contact forms of lower degree the situation is more intricate, nevertheless, we obtain a comparable local splitting.

It is noteworthy  that we obtain a compact and general expression for the splitting of a local variational morphism of any rank, which was not known in detail up to now (in \cite{FF:Natural} only a sketch of how the splitting would present in local coordinates is given).

Furthermore, an extension of the Krbek-Musilov\'a splitting, as well as the construction of a local Interior Euler operator for contact forms of lower degree is obtained {\em for any contact degree $k$}.
We work locally and the discussion of the globality of this splitting will be the subject of a separate paper.

Restricting to local splittings does not affect the main result obtained in the present paper, which is the definition, for the lower order case, of a suitable Residual operator (which even in the already known case  of $1$-contact $n$-horizontal $(n+1)$-forms is in general only locally defined).
The aim is to construct, by a recurrence formula due to Olga Rossi, higher order (local) Lepage equivalents of a given Lagrangian. 

Indeed, we recover the Krupka-Betounes Lepage equivalent for first order field theories and we obtain an elegant extension to the second order case. Again, globality properties of the latter deserve further study.

\section{Contact structure and geometric integration by parts}

Prolongations of fibered manifolds are a basic tool for the geometric formulation of the calculus of variations. We %
recall the decomposition  of pull-backs of differential forms on jet prolongations of fibered manifolds, performed by means of the jet projections and the holonomic lift of tangent vectors (a canonical construction which allows a splitting of the projection along the affine fibrations defining the contact structure of tangent vectors in two components with remarkable properties). Moreover the decomposition of forms leads to the introduction of the so-called contact forms, which reveal to be another fundamental concept in the calculus of variations. To fix notation we follow \cite{Krupka:Intro} ; other references on jet spaces are \cite{Saunders:Jets} and \cite{KMS:Natural}. 

We recall that by a \textit{fibered manifold structure} on a $C^\infty$ manifold $Y$ we mean a triplet $(Y,X,\pi)$, where $X$ is a $C^\infty$ manifold called the \textit{base} and $\pi\colon Y\longrightarrow X$ is a surjective submersion of class $C^\infty$ called the \textit{projection}. When dealing with local aspects of fibered manifolds, we will always use the so-called \textit{fibered charts} (i.e. charts adapted to the fibration).  
Let $Y$ be a fibered manifold with base $X$ and projection $\pi$, let $n=\text{dim}X$ and $m=\text{dim}Y-n$. We denote by $J^rY$, where $r\ge 0$ is any integer, the set of $r$-jets $J^r_x\gamma$ of $C^r$ \textit{sections} of $Y$ with source $x\in X$ and target $y=\gamma(x)\in Y$ (for more details on jet spaces see \cite{Saunders:Jets} and \cite{KMS:Natural}); we fix the notation $J^0Y=Y$.
For any $s$ such that $0\le s\le r$ we have surjective mappings, the \textit{canonical jet projections},  $\pi^r_s\colon J^rY\longrightarrow J^sY$ and $\pi^r\colon J^rY\longrightarrow X$, defined by  $\pi^r_s(J^r_x\gamma) = J^s_x\gamma$, $\pi^r(J^r_x\gamma) = x$.
Let $(V,\psi)$, $\psi=(x^i,y^\sigma)$, be a fibered chart on $Y$ and let $(U,\varphi)$, $\varphi=(x^i)$, be the associated chart on $X$. 

By setting $V^r=(\pi^r_0)^{-1}(V)$, a chart on the set $J^rY$ \textit{associated} with the fibered chart $(V,\psi)$ is given by  $(V^r,\psi^r)$   $\psi^r=(x^i,y^\sigma,y^\sigma_{j_1},y^\sigma_{j_1j_2},\dots,y^\sigma_{j_1j_2\dots j_r})$, with $1\le i\le n,\quad 1\le\sigma\le m$, $1\le j_1\le j_2\le\dots\le j_k\le n$, $ k=1,2,3,\dots,r$.
The set of associated charts $(V^r,\psi^r)$, such that the fibered charts $(V,\psi)$ constitute a smooth atlas on $Y$, is a smooth atlas on $J^rY$. With this smooth structure $J^rY$ is called the $r$-\textit{jet prolongation} of the fibered manifold $Y$.

Let $Y$ be a fibered manifold with base $X$ and projection $\pi$. 
Let $\Xi$ be a $\pi$-projectable vector field on $Y$, expressed in a fibered chart $(V,\psi)$, $\psi=(x^i,y^\sigma)$, by
$\Xi=\xi^i\frac{\partial}{\partial x^i}+\Xi^\sigma\frac{\partial}{\partial y^\sigma}$, then its $s$-th prolongation  $J^s\Xi$ is expressed in the associated chart $(V^s,\psi^s)$ by
\beq
J^s\Xi=\xi^i\frac{\partial}{\partial x^i}+\Xi^\sigma\frac{\partial}{\partial y^\sigma}+\sum_{k=1}^s\sum_{j_1\le j_2\le\dots\le j_k}
\Xi^\sigma_{j_1j_2\dots j_k} \frac{\partial}{\partial y^\sigma_{j_1j_2\dots j_k}},
\eeq
where $
\Xi^\sigma_{j_1j_2\dots j_k}=d_{j_k}\Xi^\sigma_{j_1j_2\dots j_{k-1}}-y^\sigma_{j_1j_2\dots j_{k-1}i}\frac{\partial\xi^i}{\partial x^{j_k}}$.

Let $J^{r+1}_x\gamma\in J^{r+1}Y$. To any tangent vector $\xi$ of $J^{r+1}Y$ at the point $J^{r+1}_x\gamma$ is assigned a tangent vector of $J^rY$ at the point $\pi^{r+1}_r(J^{r+1}_x\gamma)=J^r_x\gamma$ by
$
h\xi:=TJ^r\gamma\circ T\pi^{r+1}(\xi)$. We get a vector bundle morphism $h\colon TJ^{r+1}Y\longrightarrow TJ^rY$ over the jet projection $\pi^{r+1}_r$ called the \textit{horizontalization}, and $h\xi$ is called the \textit{horizontal component} of $\xi$. 
Let $\xi$  be given in a fibered chart $(V,\psi)$, $\psi=(x^i,y^\sigma)$ as
$
\xi=\xi^i\frac{\partial}{\partial x^i}\bigg|_{J^{r+1}_x\gamma}+\sum_{k=0}^{r+1}\sum_{j_1\le j_2\le\dots\le j_k}\Xi^\sigma_{j_1j_2\dots j_k}\frac{\partial}{\partial y^\sigma_{j_1j_2\dots j_k}}\bigg|_{J^{r+1}_x\gamma}$, then
\beq
h\xi=\xi^i d_i \byd \xi^i \big( \frac{\partial}{\partial x^i}\bigg|_{J^r_x\gamma}+\sum_{k=0}^r\sum_{j_1\le j_2\le\dots\le j_k}y^\sigma_{j_1j_2\dots j_ki}\frac{\partial}{\partial y^\sigma_{j_1j_2\dots j_k}}\bigg|_{J^r_x\gamma} \big) \,,
\eeq
where the $i$-th \textit{formal derivative operator} $d_i$ is a vector field along $\pi^{r+1}_r$ . 

We can assign to every tangent vector $\xi\in T_{J^{r+1}_x\gamma}J^{r+1}Y$ a tangent vector $p\xi\in T_{J^r_x\gamma}J^rY$ by the decomposition
$
T\pi^{r+1}_r(\xi)=h\xi+p\xi
$, where
$p\xi$ is called the \textit{contact component} of the vector $\xi$.
Then 
\beq
p\xi=\sum_{k=0}^r\sum_{j_1\le j_2\le\dots\le j_k}\Big(\Xi^\sigma_{j_1j_2\dots j_k}-y^\sigma_{j_1j_2\dots j_ki}\xi^i\Big)\frac{\partial}{\partial y^\sigma_{j_1j_2\dots j_k}}\bigg|_{J^r_x\gamma}.
\eeq

For any open set $W\subset Y$, $\Omega^r_qW$ denotes the $C^\infty$-module of $q$-forms on the open set $W^r=(\pi^r_0)^{-1}(W)$ in $J^rY$, and $\Omega^rW$ is the exterior algebra of differential forms on $W^r$. 
In order to study the structure of the components of a form $\rho\in\Omega^r_qW$, it will be convenient to introduce a \textit{multi-index notation}.
A \textit{multi-index} $I$ is an ordered $k$-tuple $I=(i_1i_2\dots i_k)$, where $k=1,2,\dots,r$ and the entries are indices such that $1\le i_1,i_2,\dots,i_k\le n$. The number $k$ is the \textit{lenght} of $I$ and is denoted by $|I|$. If $1\le j\le n$ is any integer, we denote by $Ij$ the multi-index $Ij=(i_1i_2\dots i_kj)$.

The notion of horizontalization of vectors can be used to define a morphism $h\colon\Omega^rW\longrightarrow\Omega^{r+1}W$ of exterior algebras. \\
Let $\rho\in\Omega^r_qW$, with $q\ge 1$, and $J^{r+1}_x\gamma\in W^{r+1}$. Consider the pullback  
\bEq
\label{eq:pullbk}
& &   (\pi^{r+1}_r)^\ast \rho(J^{r+1}_x\gamma)(\xi_1,\xi_2,\dots,\xi_q) =  \nonumber \\
& &   =  \rho(J^r_x\gamma)(T\pi^{r+1}_r(\xi_1),T\pi^{r+1}_r(\xi_2),\dots,T\pi^{r+1}_r(\xi_q))
\eEq
on any tangent vectors $\xi_1,\xi_2,\dots,\xi_q$ of $J^{r+1}Y$ at the point $J^{r+1}_x\gamma$. Decompose each of these vectors into the horizontal and contact components,
$
T\pi^{r+1}_r(\xi_l)=h\xi_l+p\xi_l$, 
and set 
\beq
h\rho(J^{r+1}_x\gamma)(\xi_1,\xi_2,\dots,\xi_q):=\rho(J^r_x\gamma)(h\xi_1,h\xi_2,\dots,h\xi_q).
\eeq
This formula defines a $q$-form $h\rho\in\Omega^{r+1}_qW$, while for $0$-forms
$hf:=(\pi^{r+1}_r)^\ast f$.
It follows that 
$h\rho(J^{r+1}_x\gamma)(\xi_1,\xi_2,\dots,\xi_q)$
vanishes whenever at least one of the vectors is $\pi^{r+1}$-vertical. Thus, the $q$-form $h\rho$ must be $\pi^{r+1}$-\textit{horizontal}. In particular $h\rho=0$ whenever $q\ge n+1$. The component $h\rho$ is called the \textit{horizontal component} of $\rho$. 
We say that $\rho\in\Omega^r_1W$ is \textit{contact} if $h\rho=0$.
Let us now set, 
 for $1\le k\le q$,
\beq
 & & p_k\rho(J^{r+1}_x\gamma)(\xi_1,\xi_2,\dots,\xi_q):= \\
 & & :=\frac{1}{k!(q-k)!}\sum_{\sigma\in\mathcal{P}_q}(-1)^{|\sigma|}\rho(J^r_x\gamma)(p\xi_{\sigma(1)},\dots,p\xi_{\sigma(k)},h\xi_{\sigma(k+1)},\dots,h\xi_{\sigma(q)})
\eeq
where $\mathcal{P}_q$ is the set of permutations of $q$ elements and $|\sigma|$ is the sign of the permutation $\sigma\in\mathcal{P}_q$. Note that for $k=0$, then we put $p_0\rho = h\rho$, while for $0$-forms 
$p_kf=0$, $k \ge 1$. 
In particular, given a $q$-form $\eta$ 
\beq
\eta=\sum_{s=0}^qA^{J_1}_{\sigma_1}\dots^{J_s}_{\sigma_s i_{s+1}\dots i_q}dy_{J_1}^{\sigma_1}\wedge\dots\wedge dy_{J_s}^{\sigma_s}\wedge dx^{i_{s+1}}\wedge\dots\wedge dx^{i_q} \,,
\eeq 
the $k$-contact component of $\eta$ has the chart expression
\beq
p_k\eta=B^{J_1}_{\sigma_1}\dots^{J_k}_{\sigma_k i_{k+1}\dots i_q}\omega_{J_1}^{\sigma_1}\wedge\dots\wedge\omega_{J_k}^{\sigma_k}\wedge dx^{i_{k+1}}\wedge\dots\wedge dx^{i_q},
\eeq
with 
$\omega^\sigma_J=dy^\sigma_J-y^\sigma_{Ji}dx^i$,
and 
\beq
B^{J_1}_{\sigma_1}\dots^{J_k}_{\sigma_k i_{k+1}\dots i_q}=\sum_{s=k}^q\binom{s}{k}A^{J_1}_{\sigma_1}\dots^{J_kJ_{k+1}}_{\sigma_k\sigma_{k+1}}\dots^{J_s}_{\sigma_s [i_{s+1}\dots i_q}y^{\sigma_{k+1}}_{J_{k+1}i_{k+1}}\dots y^{\sigma_s}_{J_s i_s]}
\eeq
where the antisymmetrization in the right hand side of the last equation is performed only on the indices $i_{k+1}\dots i_si_{s+1}\dots i_q$. 

For any $\rho\in\Omega^r_qW$,  $q\ge 0$,  the \textit{canonical decomposition} of the form $\rho$ is given as 
\beq
(\pi^{r+1}_r)^\ast\rho=h\rho+p_1\rho+p_2\rho+\dots+p_q\rho.
\eeq
We can see that the canonical decomposition of forms gives rise to the splitting of the pull-back of the exterior derivative 
\beq
(\pi^{r+2}_r)^\ast d\rho=d_H\rho+d_C\rho 
:=\sum_{k=0}^qp_k dp_k\rho + \sum_{k=0}^qp_{k+1}dp_k\rho \,,,
\eeq
and characterized by the identities $d_H\circ d_H=0$, $d_C\circ d_C=0$, $d_C\circ d_H=-d_H\circ d_C$; furthermore,  if $\rho$ is a $q$-form and $\eta$ is an $s$-form, both on $J^rY$, then
\beq
& &  d_H(\rho\wedge\eta)=d_H\rho\wedge(\pi^{r+2}_r)^\ast\eta+(-1)^q(\pi^{r+2}_r)^\ast\rho\wedge d_H\eta \\
& &  d_C(\rho\wedge\eta)=d_C\rho\wedge(\pi^{r+2}_r)^\ast\eta+(-1)^q(\pi^{r+2}_r)^\ast\rho\wedge d_C\eta.
\eeq

\subsection{The interior Euler operator}

We recall some technical features of the interior Euler operator, seen as a tool which allows to pass in a univocal way from equivalence classes of local differential forms in the variational sequence, to (global) differential forms in the representation sequence. First we shortly recall the \textit{finite order} variational sequence as introduced by Krupka in \cite{Krupka:Variational}. A complete description of this subject involves some topics of sheaf theory and sheaf cohomology; however, since our purpose is to make direct calculations on the representation of the variational sequence, we just refer to \cite{Krupka:Intro} for more details about those aspects. Then we shortly recall the notion of \textit{Lie derivative} of forms 
with respect to a `vector field along a map' 
and some results about integration by parts formulae which lead directly to the definition of the interior Euler operator. For more details and other related topics we refer to \cite{KrbekMusi:Representation} and \cite{Palese:VariationalSeq}. \\

Let $\Omega^r_q$, $q\ge 0$, be the direct image of the sheaf of smooth $q$-forms over $J^rY$ by the jet projection $\pi^r_0$. We denote by
\beq
\Omega^r_{q,\text{c}}=
\begin{cases}
\text{ker}\,p_0	\quad	 \text{for}\,\, 1\le q\le n, \\
\text{ker}\,p_{q-n}	\quad \text{for}\,\, n+1\le q\le \text{dim}J^rY
\end{cases}
\eeq
the sheaf of contact $q$-forms, if $q\le n$, or the sheaf of strongly contact $q$-forms, if $n+1\le q\le \text{dim}J^rY$. 

We set
\beq
\Theta^r_q=\Omega^r_{q,\text{c}}+d\Omega^r_{q-1,\text{c}}
\eeq
where $d\Omega^r_{q-1,\text{c}}$ is the image sheaf of $\Omega^r_{q-1,\text{c}}$ by the exterior derivative $d$. The sequence of sheaves 
$
\{0\}\to\Theta^r_1\to\dots\to\Theta^r_n\to\Theta^r_{n+1}\to\dots\to\Theta^r_P\to\{0\}
$,
with $d$ the  exterior derivatives and $P$ being the maximal nontrivial degree, is an exact subsequence of the de Rham sequence. The  acyclic resolution of the constant sheaf $\mathbb{R}_Y$ over $Y$, given by 
$
\{0\}\to\mathbb{R}_Y \to\Omega^r_*/\Theta^r_*$, is called the \textit{variational sequence of order} $r$. We denote the quotient mappings as $E^r_q\colon [\rho]\in\Omega^r_q/\Theta^r_q\longrightarrow E^r_q([\rho])=[d\rho]\in\Omega^r_{q+1}/\Theta^r_{q+1}$ Note that, in particular, the mappings $E^r_n$ and $E^r_{n+1}$ correspond to the Euler-Lagrange mapping and to the Helmholtz-Sonin mapping of calculus of variations, respectively.

\begin{defn}
Let $(V,\psi)$, $\psi=(x^i,y^\sigma)$, be a fibered chart on $Y$ and let $\rho$ be a differential $q$-form on $J^rY$.
The \textit{Lie derivative of a $q$-form $\rho$ on with respect to a vector field $h\Xi$ along the map $\pi^{r+1}_r$} is given by
\beq
\pounds_{h\Xi}^{\pi^{r+1}_r}\rho=
(\pi^{r+1}_r)^\ast_{h\Xi} ( h\Xi \lrcorner d\rho)+d (\pi^{r+1}_r)^\ast_{h\Xi} (h\Xi \lrcorner\rho) \,.
\eeq\end{defn}
\noindent Here $(\pi^{r+1}_r)^\ast_{h\Xi}$ is a pull-back defined according to \cite{KrbekMusi:Representation}.

In particular, let $d_i$ be the $i$-th \textit{formal derivative operator} seen as a (horizontal) vector field along a map. We have 
$\pounds^{\pi^{r+1}_r}_{\text{d}_i}dx^j=0 $, $
\pounds^{\pi^{r+1}_r}_{\text{d}_i}dy^\sigma_J=dy^\sigma_{Ji} $, $ \pounds^{\pi^{r+1}_r}_{\text{d}_i}\omega^\sigma_J=\omega^\sigma_{Ji}$, while for $f$ a zero form, we have 
$\pounds^{\pi^{r+1}_r}_{\text{d}_i}f=\frac{\partial f}{\partial x^i}+\sum_{|J|=0}^ry^\sigma_{Ji}\frac{\partial f}{\partial y^\sigma_J} = d_i f $.

Accordingly, by a slight abuse of notation, we will use at any degree the symbol $\text{d}_i=\pounds^{\pi^{r+1}_r}_{\text{d}_i}$, and we will call it the \textit{total derivative of forms} with respect to the coordinate $x^i$.

We recall that the total derivatives of forms enjoy the following properties
\begin{enumerate} 
\item the form $\text{d}_H\rho$ can be locally decomposed as
\beq
\text{d}_H\rho=(-1)^q\text{d}_i\rho\wedge dx^i
\eeq
\item the Leibniz rule holds for total derivatives of the exterior product of forms $\rho$ and $\eta$
\beq
\text{d}_i(\rho\wedge\eta)=\text{d}_i\rho\wedge\eta+\rho\wedge\text{d}_i\eta
\eeq
\item let $(\bar{V},\bar{\psi})$, $\bar{\psi}=(\bar{x}^j,\bar{y}^\nu)$, be a fibered chart on $Y$, such that $V\cap\bar{V}\ne\emptyset$ and let $\bar{\text{d}}_j$ be the total derivative with respect to the coordinate $\bar{x}^j$. Then the  transformation rule $\text{d}_i\rho=\frac{\partial\bar{x}^j}{\partial x^i}\bar{\text{d}}_j\rho
$ holds.
\item the total derivatives commute, \ie
\beq
\text{d}_i\text{d}_j\rho=\text{d}_j\text{d}_i\rho.
\eeq
\end{enumerate}
This last property allows us to use the notation 
$\text{d}_J=\text{d}_{j_s}\circ\dots\circ\text{d}_{j_1} $, where $J=(j_1\dots j_s)$ is a multi-index.

We consider the generalization of the integration by parts to differential forms based on the above concept of total derivative of forms  and due to \cite{KrbekMusi:Representation}. 

Let $(V,\psi)$, $\psi=(x^i,y^\sigma)$, be a fibered chart on $Y$ and $\rho\in\Omega^r_{n+k}V$ a form. Let $\text{p}_k\rho$ be expressed as
\beq
\label{eq:IE}
\text{p}_k\rho=\sum_{|J|=0}^r\omega^\sigma_J\wedge\eta^J_\sigma.
\eeq
Then there exists the decomposition
\bEq
\label{eq:32}
p_k\rho = \mathcal{I}(\rho) + p_kdp_k\mathcal{R}(\rho)
\eEq
where $\mathcal{I}$  is the \textit{interior Euler operator},  $\mathcal{R}$  is  the \textit{Residual  operator}, and $\mathcal{R}(\rho)$ is a local $k$-contact $(n+k-1)$-form.

There exists a unique decomposition as above such that $\mathcal{I}$ is $\mathbb{R}$-linear, which is therefore globally defined.
In local coordinates we have
\bEq
\label{eq:IEO}
\mathcal{I}\colon\Omega^r_{n+k}W\ni\rho\longrightarrow\mathcal{I}(\rho)=\frac{1}{k}\omega^\sigma\wedge\sum_{|J|=0}^r(-1)^{|J|}\text{d}_J(\frac{\partial}{\partial y^\sigma_J}\lrcorner\,p_k\rho)\in\Omega^{2r+1}_{n+k}W \,.
\eEq

Let $W\subset Y$ be an open set and let $\rho\in\Omega^r_{n+k}W$, $1\le k\le\text{dim}J^rY-n$, be a form. Then the following intrinsic properties uniquely characterize the interior Euler operator:
\begin{enumerate}
\item[(a)] $(\pi^{2r+1}_r)^\ast\rho-\mathcal{I}(\rho)\in\Theta^{2r+1}_{n+k}W$;
\item[(b)] $\mathcal{I}(p_kdp_k\mathcal{R}(\rho))=0$;
\item[(c)] $\mathcal{I}^2(\rho)=(\pi^{4r+3}_{2r+1})^\ast\mathcal{I}(\rho)$;
\item[(d)] $\text{ker}(\mathcal{I})=\Theta^r_{n+k}W$.
\end{enumerate}

\subsection{Variational morphisms and canonical splittings}

We now recall shortly  the definition and the basic properties of variational morphisms; see \cite{FF:Natural}. In view of a comparison with the Krbek-Musilov\'a geometric integration by parts, we discuss their algorithmic splitting properties, which correspond to the possibility of performing a global and \textit{covariant} integration by parts. We distinguish the case of codegree $s=0$ from the case $0<s\le n$ and include some results about the uniqueness properties of the aforementioned splittings.

\begin{defn}
\label{defin:morph}
Let $\mathcal{E} = (E,X, \tilde{\pi}, \mathbb{R}^l)$ be a vector bundle and $\pi\colon Y\longrightarrow X$ an arbitrary fiber bundle, both over $X$, with $\text{dim}X=n$,
and let $A_{q}(X)$ denote the bundle of $q$-forms on $X$.
Let $t$, $r$ and $s$ be integers.
 A bundle morphism
\beq
\mathbb{V}\colon J^{t}Y\longrightarrow (J^{r}\mathcal{E})^\ast\otimes A_{n-s}(X)
\eeq
is called a \textit{variational} $\mathcal{E}$\textit{-morphism on} $Y$. The (minimal) integer $t$ is called the \textit{order} of $\mathbb{V}$, $r$ is called the \textit{rank} and $(n-s)$ is called the \textit{degree} of $\mathbb{V}$ (being $s$ the \textit{codegree}).
\end{defn}

A fibered connection on $\mathcal{E}$ (\ie a linear connection $\Gamma^a_{bi}$ on $X$ and a connection $\Gamma^A_{Bi}$ on $\mathcal{E}$) induces on 
$(J^r\mathcal{E})^\ast\otimes A_{n-s}(X)$ a set of local fibered coordinates $(x^i;\hat{v}^{i_1\dots i_s}_A,\dots,\hat{v}^{i_1\dots i_sj_1\dots j_r}_A)$ so that 
a variational morphism $\mathbb{V}$ can be locally given there as
\beq
<\mathbb{V}|J^r\Xi>=\frac{1}{s!}\Bigl[\hat{v}^{i_1\dots i_s}_A\hat{\Xi}^A+\hat{v}^{i_1\dots i_sj}_A\hat{\Xi}^A_j+\dots+\hat{v}^{i_1\dots i_sj_1\dots j_r}_A\hat{\Xi}^A_{j_1\dots j_r}\Bigr]\otimes ds_{i_1\dots i_s}
\eeq
where $ds_{i_1\dots i_s}=\frac{\partial}{\partial x^{i_s}}\lrcorner\,\dots\,\lrcorner\,\frac{\partial}{\partial x^{i_1}}\lrcorner\,ds$, if $ds$ is the volume density on the base manifold $X$.
Each coefficient $\mathbb{V}_m=\frac{1}{s!}\hat{v}^{i_1\dots i_sj_1\dots j_m}_A$ of order  $0\le m\le r$ is the coefficient of a global variational morphism, called the $m$\textit{-rank term of} $\mathbb{V}$ (if  $m=r$ it is called the \textit{highest rank term of} $\mathbb{V}$).

Let now $\mathbb{Q}\colon J^{t}Y\longrightarrow A_{n-s}(X)$ be a morphism of rank $r=0$. The \textit{divergence} of $\mathbb{Q}$ is the variational morphism $\text{Div}(\mathbb{Q})\colon J^{t+1}Y\longrightarrow A_{n-s+1}(M)$ such that
\beq
\text{Div}(\mathbb{Q})\circ J^{t+1}\sigma=d(\mathbb{Q}\circ J^t\sigma)
\eeq
for each section $\sigma\colon X\longrightarrow Y$.

Variational morphisms admit canonical and algorithmic splittings corresponding to global and covariant integration by parts.
We distinguish two cases.

\begin{itemize}

\item the case of codegree $s=0$.

Let $\mathbb{V}\colon J^{t}Y\longrightarrow (J^{r}\mathcal{E})^{\ast}\otimes A_n(X)$ be a variational $\mathcal{E}$-morphism of codegree $s=0$. Then we can define two global variational $\mathcal{E}$-morphisms
\beq
 & & \mathbb{E}\equiv\mathbb{E}(\mathbb{V})\colon J^{t+r}C\longrightarrow\mathcal{E}^\ast\otimes A_n(X) \\
 & & \mathbb{T}\equiv\mathbb{T}(\mathbb{V})\colon J^{t+r-1}C\longrightarrow(J^{r-1}\mathcal{E})^\ast\otimes A_{n-1}(X)
\eeq
such that the following splitting property holds true:
\bEq\label{eq:split}
<\mathbb{V}|J^r\Xi>=<\mathbb{E}|\Xi>+\text{Div}(<\mathbb{T}|J^{r-1}\Xi>)
\eEq
for each section $\Xi$ of $\mathcal{E}$. The variational morphism $\mathbb{E}$ and $\mathbb{T}$ are called the volume part and the boundary part of $\mathbb{V}$, respectively.\end{itemize}

In particular, 
we locally have:
\beq
<\mathbb{E}|\Xi>=\Big[\Big(\hat{v}_A-\nabla_{j_1}\hat{v}_A^{j_1}+\dots +(-1)^r\nabla_{j_1\dots j_r}\hat{v}_A^{j_1\dots j_r}\Big)\Xi^A\Big]\otimes ds
\eeq
and
\beq
<\mathbb{T}|J^{r-1}\Xi>=\Bigl[\hat{t}^i_A\hat{\Xi}^A+\hat{t}^{ij_1}_A\hat{\Xi}^A_{j_1}+\dots+\hat{t}^{ij_1\dots j_{r-1}}_A\hat{\Xi}^A_{j_1\dots j_{r-1}}\Bigr]\otimes ds_i
\eeq
where the coefficients of $\mathbb{T}$ are given by the recurrence relations
\bEq \label{eq:Tcomponents}
& &  \hat{t}^{ij_1\dots j_{r-1}}_A=\hat{v}^{ij_1\dots j_{r-1}}_A   \nonumber \\ 
& &  \hat{t}^{ij_1\dots j_{r-2}}_A=\hat{v}^{ij_1\dots j_{r-2}}_A-\nabla_l\hat{t}^{lij_1\dots j_{r-2}}_A \\ 
& &  \dots \nonumber \\ 
& &  \hat{t}^i_A=\hat{v}^i_A-\nabla_l\hat{t}^{li}_A \,. \nonumber
\eEq

A similar splitting formula can be obtained for variational morphisms of higher codegree. The expressions in coordinates are in this case more complicated due to the presence of some antisymmetries in the indices of the local coefficients; see in particular \cite{FF:Natural}, p. $161$--$162$. 
The concept of a reduced morphism (with respect to a fibered connection) is shown to be necessary.

\begin{defn} \label{def:1.5}
Let $\mathbb{V}\colon J^{t}Y\longrightarrow (J^{r}\mathcal{E})^{\ast}\otimes A_{n-s}(X)$ be a variational morphism.  \\
Let  $\mathbb{V}_m=\frac{1}{s!}(\hat{v}_A^{i_1\dots i_sj_1\dots j_m})\otimes ds_{i_1\dots i_s}$ be the coefficient of its term of rank $0\le m\le r$. 

The term $\mathbb{V}_m$ is said to be \textit{reduced with respect to the fibered connection} $(\Gamma^a_{bi},\Gamma^A_{Bi})$ if $\hat{v}_A^{[i_1\dots i_sj_1]j_2\dots j_m}=0$. The variational morphism $\mathbb{V}$ is \textit{reduced} if all its terms are reduced.
\end{defn}
Notice that when $n=\text{dim}(X)=1$, \eg in the case of Mechanics, all variational morphisms are reduced.
However, the property of being reduced in general depends on the choice of the a fibered connection, whenever the rank is at least two.

\begin{itemize}

\item the case of codegree $s\ge 1$.

Let  now $\mathbb{V}\colon J^{t}Y\longrightarrow (J^{r}\mathcal{E})^{\ast}\otimes A_{n-s}(X)$ be a global variational $\mathcal{E}$-morphism of codegree $s\ge 1$. Then we can define two global variational $\mathcal{E}$-morphisms
\beq
& & \mathbb{E}\equiv\mathbb{E}(\mathbb{V})\colon J^{t+r}C\longrightarrow(J^r\mathcal{E})^\ast\otimes A_{n-s}(X) \\
& &  \mathbb{T}\equiv\mathbb{T}(\mathbb{V})\colon J^{t+r-1}C\longrightarrow(J^{r-1}\mathcal{E})^\ast\otimes A_{n-s-1}(X)
\eeq
where $\mathbb{E}$ is a {\em reduced} variational morphism and such that the following holds true:
\bEq
\label{eq:splitcodeg}
<\mathbb{V}|J^r\Xi>=<\mathbb{E}|J^r\Xi>+\text{Div}(<\mathbb{T}|J^{r-1}\Xi>)
\eEq
for each section $\Xi$ of $\mathcal{E}$. Only a  sketch of the expression in local coordinates of the volume part  $\mathbb{E}$ and the  boundary part  $\mathbb{T}$  can be found in \cite{FF:Natural}.

\end{itemize}

Let a fibered connection be fixed; the volume part is uniquely determined, while the boundary part is determined modulo a divergenceless term.
When $r\ge 2$ one can proceed by further splitting 
\bEq
\label{eq:33}
<\mathbb{T}|J^{r-1}\Xi>=<\mathbb{S}|J^{r-1}\Xi>+\text{Div}(<\mathbb{Q}|J^{r-2}\Xi>)\,,
\eEq
where the variational morphism $\mathbb{S}\colon J^{t+2r-2}Y\longrightarrow (J^{r-1}\mathcal{E})^{\ast}\otimes A_{n-s-1}(X)$ is reduced by construction and uniquely determined.

\section{Comparison of the two approaches and new results}

We present some original results which clarify the similarities and differences between the two integration by parts methods described above.

The basic idea is that $1$-contact forms of degree $n+1$ can be seen as variational morphisms and, {\em viceversa}, to each variational morphism a $1$-contact form of degree $n+1$ can be associated. In Proposition \ref{prop:Volume}, we prove the equivalence of decompositions \eqref{eq:32} 
and \eqref{eq:split} for $1$-contact $(n+1)$-forms (which we shall call \textit{top forms}, because they are of the highest horizontal degree), seen as variational morphisms of codegree $s=0$.

\subsection{Contact forms as (local) variational morphisms}

\label{sec:formsmorph}
Consider an arbitrary bundle $\pi\colon Y\longrightarrow X$ with $n=\text{dim}X$. Let $U\subseteq X$ be an open subset and let $W=\pi^{-1}(U)$ be the "tube" over $U$. Consider $\rho\in\Omega^r_qW$ a $1$-contact $q$-form on $W^r=(\pi^r_0)^{-1}(W)$, with $q\le n+1$. Then, if $(W^r,\psi^r)$ is a local chart on $J^rY$ associated with the fibered chart $(W,\psi)$, $\psi=(x^i,y^\sigma)$ on $Y$, we can write
\beq
p_1\rho=\sum_{|J|=0}^r\omega^\sigma_J\wedge\eta^J_\sigma\,\in\Omega^{r+1}_q W
\eeq
where $\eta^J_\sigma$ are horizontal $(q-1)$-forms defined on $W^{r+1}$ and thence can be expressed as
\beq
\eta^J_\sigma=A^{i_1\dots i_sJ}_\sigma(J^{r+1}y)ds_{i_1\dots i_s}, \qquad s=n-(q-1). 
\eeq

Now, considering the vector bundle $V(W)$ whose sections are vertical vector fields over $W \to U$ and recalling that $J^rV(W)\cong V(J^rW)$, we can define according to Definition \ref{defin:morph} a variational morphism $\mathbb{V}_\rho\colon J^{r+1}W\longrightarrow(J^rV(W))^\ast\otimes A_s(U)$ such that:
\beq
 & & <\mathbb{V}_\rho|J^r\Xi>=J^r\Xi\lrcorner p_1\rho=\Big[\sum_{|J|=0}^rA^{i_1\dots i_sJ}_\sigma\Xi^\sigma_J\Big]\otimes ds_{i_1\dots i_s}= \\
& & =\Big[A^{i_1\dots i_s}_\sigma\Xi^\sigma+A^{i_1\dots i_sj_1}_\sigma\Xi^\sigma_{j_1}+\dots +A^{i_1\dots i_sj_1\dots j_r}_\sigma\Xi^\sigma_{j_1\dots j_r}\Big]\otimes ds_{i_1\dots i_s}
\eeq
for every vertical vector field $\Xi\colon W\longrightarrow V(W)$.

The advantage of this approach consists in the possibility of working on contact forms (although only in the particular case of $1$-contact forms of degree at most $n+1$) using the tools of the theory of variational morphisms and returning back to forms at the end of the manipulation.

It appears that the above identification of $1$-contact forms with variational morphisms holds true up to $(n+1)$-forms and, indeed, this could be related to the non-uniqueness of the \textit{source forms} providing the so-called \textit{Helmholtz conditions}. In fact, as discussed in \cite{Palese:VariationalSeq}
(pag. 32), this feature appears for $k$-contact $n$-horizontal $(n+k)$-forms with $k\ge 2$.

\subsection{Comparison for top forms}
\label{sec:topforms}

In this section we directly compare the two integration by parts procedures. 

As a first step, in the following proposition we prove that the splitting \eqref{eq:split} of $p_1\rho$, seen as a variational morphism $\mathbb{V}_\rho$ of codegree $s=0$,
\beq
<\mathbb{V}_\rho|J^r\Xi>=<\mathbb{E}|\Xi>+\text{Div}(<\mathbb{T}|J^{r-1}\Xi>)
\eeq
and the decomposition
\beq
p_1\rho=\mathcal{I}(\rho)+d_H\mathcal{R}(\rho)
\eeq
give the same terms.

\begin{prop}
\label{prop:Volume}
Given $\rho\in\Omega^r_{n+1}W$ a $1$-contact $(n+1)$-form.
For every section $\Xi\colon W\longrightarrow V(W)$,
\beq
<\mathbb{V}_\rho|J^r\Xi> =
J^r\Xi\lrcorner p_1\rho
\eeq
 and 
\beq
<\mathbb{E}|\Xi>=J^r\Xi\lrcorner\,\mathcal{I}(\rho) \,, \quad 
\text{Div} (<\mathbb{T}|J^{r-1}\Xi>)=J^r\Xi\lrcorner\,d_H\mathcal{R}(\rho).
\eeq
\end{prop} 
\begin{proof}
\textit{Step 1.} The fact that $<\mathbb{E}|\Xi>=J^r\Xi\lrcorner\,\mathcal{I}(\rho)$ follows directly from the definition of each side of the equation. In fact, since $W$ is a single coordinate domain, we can always choose the fibered connection whose coefficients are all null, then the covariant derivatives reduce to total derivatives and the variational morphism $\mathbb{E}$ takes the form:
\beq
<\mathbb{E}|\Xi> = [(A_\sigma-d_{j_1}A_\sigma^{j_1}+\dots +(-1)^rd_{j_1\dots j_r}A_\sigma^{j_1\dots j_r})\Xi^\sigma ]\otimes ds.
\eeq
On the other hand, from the definition of the interior Euler operator, we have
\bEq \label{topInterior}
\mathcal{I}(\rho) = \omega^\sigma\wedge\sum_{|J|=0}^r(-1)^{|J|}\text{d}_J\eta^J_\sigma =  \omega^\sigma\wedge\Big(\sum_{|J|=0}^r(-1)^{|J|}d_JA^J_\sigma\Big)ds =  \nonumber \\
= \omega^\sigma\wedge (A_\sigma-d_{j_1}A^{j_1}_\sigma+\dots +(-1)^rd_{j_1\dots j_r}A^{j_1\dots j_r}_\sigma )ds.
\eEq

\textit{Step 2.} In order to compare $J^r\Xi\lrcorner\,d_H\mathcal{R}(\rho)$ with $\text{Div}(<\mathbb{T}|J^{r-1}\Xi>)$ we need to compute explicitly the Residual operator $\mathcal{R}(\rho)$. 
We now write
$
p_1\rho=\sum_{|I|=0}^r\text{d}_I\big(\omega^\sigma\wedge\xi^I_\sigma\big)
$, 
where
\beq
\xi^I_\sigma =\sum_{|J|=0}^{r-|I|}(-1)^{|J|}\binom{|J|+|I|}{|J|}\text{d}_J\eta^{IJ}_\sigma= \bigg(\sum_{|J|=0}^{r-|I|}(-1)^{|J|}\binom{|J|+|I|}{|J|}\text{d}_JA^{IJ}_\sigma\bigg) ds \,.
\eeq
The summand $\omega^\sigma\wedge\xi_\sigma$ is the interior Euler operator, so we consider only the remaining terms $\sum_{|I|=1}^r\text{d}_I\big(\omega^\sigma\wedge\xi^I_\sigma\big)$. Each form $\omega^\sigma\wedge\xi^I_\sigma$ is a $1$-contact $(n+1)$-form and thence can be recast as $\omega^\sigma\wedge\xi^I_\sigma=\chi^I\wedge ds$, where $\chi^I$ is a $1$-contact $1$-form locally given as 
\beq
\chi^I=(\sum_{|J|=0}^{r-|I|}(-1)^{|J|}\binom{|J|+|I|}{|J|}\text{d}_JA^{IJ}_\sigma )\omega^\sigma.
\eeq
Finally, the   Residual operator is defined by
\beq
& & \mathcal{R}(\rho)=\sum_{|I|=0}^{r-1}(-1)^{1}\text{d}_I\chi^{iI}\wedge ds_i= \\
& & =\sum_{|I|=0}^{r-1}-\text{d}_I\bigg(\sum_{|J|=0}^{r-|iI|}(-1)^{|J|}\binom{|J|+|iI|}{|J|}\omega^\sigma\text{d}_JA^{iIJ}_\sigma\bigg)\wedge ds_i.
\eeq
Now, from the second equation above, we compute the coefficients of the forms $\omega^\sigma_L\wedge ds_i$ according to the length of the multi-index $L$:
\begin{description}
\item[$|L|=r-1$.] The only contribution to this coefficient comes from setting $|I|=r-1$ and applying all the total derivatives of forms $\text{d}_I$ to $\omega^\sigma$, thus obtaining
as coefficient $A^{iL}_\sigma$.
\item[$|L|=r-2$.] One contribution comes from setting $|I|=r-2$ and applying all the total derivatives of forms $\text{d}_I$ to $\omega^\sigma$, getting the coefficient
\beq
-(A^{iL}_\sigma-rd_lA^{liL}_\sigma),
\eeq
another contribution comes from setting $|I|=r-1$ and applying $r-2$ total derivatives to $\omega^\sigma$ and one to $A^{iI}_\sigma$, thus obtaining
\beq
-\binom{r-1}{1}d_lA^{liL}_\sigma,
\eeq
summing together these two terms, we get the coefficient
\beq
-(A^{iL}_\sigma-d_lA^{liL}_\sigma).
\eeq
\item[$|L|=r-3$.] We can take $|I|=r-3$ and apply all the total derivatives to $\omega^\sigma$, getting a term
\beq
- A^{iL}_\sigma+\binom{r-1}{1}d_lA^{liL}_\sigma-\binom{r}{2}d_{lk}A^{lkiL}_\sigma \,,
\eeq
another contribution comes from setting $|I|=r-2$ and applying only $r-3$ derivatives to $\omega^\sigma$, getting a term
\beq
-\binom{r-2}{1}d_lA^{liL}_\sigma+\binom{r}{1}\binom{r-2}{1}d_{lk}A^{lkiL}_\sigma \,,
\eeq
finally, we can take $|I|=r-1$ and apply $r-3$ derivatives to $\omega^\sigma$, getting a term
\beq
-\binom{r-1}{2}d_{lk}A^{lkiL}_\sigma,
\eeq
summing all together these contributions we obtain the following coefficient
\beq
-(A^{iL}_\sigma-d_lA^{liL}_\sigma+d_{lk}A^{lkiL}_\sigma) \,. 
\eeq 
In general, we thus get 
\item[$|L|=r-m$, $1\leq m \leq r$.]
\beq
-(A^{iL}_\sigma-d_{j_1}A^{iL j_1}_\sigma+   \dots (-1)^{m-1}  d_{j_1 j_2j_3\dots j_{m-1}}A^{iL j_1 j_2j_3\dots j_{m-1}}_\sigma ) \,. 
\eeq
\end{description}
It appears clear then  that the coefficients of the forms $\omega^\sigma_L\wedge ds_i$ are defined by the same recurrence relations \eqref{eq:Tcomponents} which express the components $t^{iJ}_\sigma$ of the variational morphism $\mathbb{T}$, except for the sign. Thence we can write for the Krbek-Musilov\'a Residual operator
\beq
\mathcal{R}(\rho)=\sum_{|J|=0}^{r-1}-t^{iJ}_\sigma\omega^\sigma_J\wedge ds_i.
\eeq
Therefore, $d_H\mathcal{R}(\rho)=p_1dp_1\mathcal{R}(\rho)
=\sum_{|J|=0}^{r-1}(d_it^{iJ}_\sigma\omega^\sigma_J\wedge ds+t^{iJ}_\sigma\omega^\sigma_{Ji}\wedge ds )$, and since $\Xi^\sigma_{Ji}=d_i\Xi^\sigma_J$, we finally get 
\beq
& & J^r\Xi\lrcorner\,d_H\mathcal{R}(\rho)=\sum_{|J|=0}^{r-1}\Big(d_it^{iJ}_\sigma\Xi^\sigma_J+t^{iJ}_\sigma\Xi^\sigma_{Ji}\Big)\wedge ds= \\ 
& & = \sum_{|J|=0}^{r-1}\Big(d_it^{iJ}_\sigma\Xi^\sigma_J+t^{iJ}_\sigma d_i\Xi^\sigma_J\Big)\wedge ds \\
& &  =d_i\Big(\sum_{|J|=0}^{r-1}t^{iJ}_\sigma\Xi^\sigma_J\Big)\wedge ds= \text{Div}(<\mathbb{T}|J^{r-1}\Xi>).
\eeq
\end{proof}

\bRm 
Note that, since $d_H ds_{i_1\dots i_s}=0$,  the latter equation can be generalized to hold true for every $0\le s\le n$.
Indeed, let $\mathbb{T}\colon J^tY\longrightarrow (J^{r-1}V(Y))^\ast\otimes A_{n-s}(X)$ be a variational morphism according to Definition \ref{defin:morph}, with
\beq
<\mathbb{T}|J^{r-1}\Xi> = (\sum_{|J|=0}^{r-1}t^{i_1\dots i_sJ}_\sigma\Xi^\sigma_J)\wedge ds_{i_1\dots i_s}
\eeq
for any vertical vector field $\Xi$ over $Y$. Then there is a correspondence between $\mathbb{T}$ and a $(n-s)$-horizontal $1$-contact $(n-s+1)$-form $\tilde{\mathcal{R}}$ such that
\beq
\text{Div}(<\mathbb{T}|J^{r-1}\Xi>)=J^r\Xi\lrcorner\,d_H\tilde{\mathcal{R}}
\eeq
where $\tilde{\mathcal{R}}$ is defined by
\beq
\tilde{\mathcal{R}} = (\sum_{|J|=0}^{r-1}-t^{i_1\dots i_sJ}_\sigma\omega^\sigma_J )\wedge ds_{i_1\dots i_s} \,.
\eeq
In spite of the suggestive notation, $\tilde{\mathcal{R}}$ does {\em not} deal with 
a Residual operator at this stage. 
\eRm

\section{The main result: comparison for lower degree forms}    \label{sec:lowdeg}

The more intricate case of $1$-contact $(n-s+1)$-forms, seen as variational morphisms of codegree $0<s\le n$, is discussed. 

First we show that, in general, for $k$-contact $(n-s)$-horizontal $(n-s+k)$-forms, with an adequate manipulation, it is possible to obtain a decomposition analogous to \eqref{eq:32}. 
In Proposition \ref{prop:div} we characterize the boundary term by a local differential operator which {\em extends to the case of $k$-contact  forms of lower degree the Krbek-Musilov\'a's  Residual operator}. It is defined for forms of lower degree {\em of any order of contactness $k$}.

In Proposition \ref{prop:splitlike} we show that, when we restrict to $k=1$, this decomposition is indeed equivalent to the application of a ``canonical splitting''-like algorithm to the corresponding variational morphism. We also define a suitable {\em extension of the interior Euler operator} to $1$-contact forms of lower degree.
In this specific case, we also provide some examples, in particular  for $r=1 \,, 0 \le s <n$ and for $r=2 \, , s=1$, where we compare this decomposition with \eqref{eq:splitcodeg}; the latter example shows that the difference between the respective  boundary terms is indeed a local divergence, dealing with the splitting \eqref{eq:33}.
\medskip

Let  $\pi\colon Y\longrightarrow X$, and $U\subseteq X$ an open subset and $W=\pi^{-1}(U)$. Let $\rho\in\Omega^r_{n-s+k}W$ be a $(n-s)$-horizontal $k$-contact $(n-s+k)$-form defined on the $r$-order jet prolongation $W^r$ of $W$. In a local fibered chart $\psi^r=(x^i,y^\sigma,y^\sigma_I)$, we can write:
\bEq \label{expression of prho}
p_k\rho=\sum_{|J|=0}^r\omega^\sigma_J\wedge\eta^J_\sigma\,\in\Omega^{r+1}_{n-s+k}W
\eEq
where $\eta^J_\sigma$ are local $(n-s)$-horizontal $(k-1)$-contact $(n-s+k-1)$-forms defined on $W^{r+1}$.
Again, according to \cite{KrbekMusi:Representation}, we can rewrite 
\beq 
p_k\rho =\sum_{|J|=0}^r\omega^\sigma_J\wedge\eta^J_\sigma=\sum_{|I|=0}^r\text{d}_I\big(\omega^\sigma\wedge\xi^I_\sigma\big)
 =\omega^\sigma\wedge\xi_\sigma+\sum_{|I|=1}^r\text{d}_I\big(\omega^\sigma\wedge\xi^I_\sigma\big)
\eeq
where
\bEq \label{KrMu lemma}
\xi^I_\sigma=\sum_{|J|=0}^{r-|I|}(-1)^{|J|}\binom{|J|+|I|}{|J|}d_J\eta^{IJ}_\sigma.
\eEq

In analogy with \textit{Step 2} of the above Proposition (which can be recovered for $k=1\,, s=0$)
we work out the term $\sum_{|I|=1}^r\text{d}_I\big(\omega^\sigma\wedge\xi^I_\sigma\big)$. 
Each form $\omega^\sigma\wedge\xi^I_\sigma$ is a $(n-s)$-horizontal $k$-contact $(n-s+k)$-form, thus it can be recast in the following manner:
\bEq
\label{eq:sost}
\omega^\sigma\wedge\xi^I_\sigma=\chi^{i_1\dots i_sI}\wedge ds_{i_1\dots i_s}
\eEq
where the $\chi^{i_1\dots i_sI}$ are local $k$-contact $k$-forms. 
Renaming the multi-index $I$ and extracting an antisymmetric part, we obtain:
\bEq
\label{eq:dec1} 
& &  \sum_{|I|=1}^r\text{d}_I\big(\omega^\sigma\wedge\xi^I_\sigma\big)
= \sum_{|I|=0}^{r-1}\text{d}_i\text{d}_I\chi^{i_1\dots i_siI}\wedge ds_{i_1\dots i_s}= 
 \\
& &  =\sum_{|I|=0}^{r-1}\text{d}_i\text{d}_I\Big(\chi^{i_1\dots i_siI}-\chi^{[i_1\dots i_si]I}\Big)\wedge ds_{i_1\dots i_s}+\text{d}_i\sum_{|I|=0}^{r-1}\text{d}_I\chi^{[i_1\dots i_si]I}\wedge ds_{i_1\dots i_s}. \nonumber 
\eEq

\begin{lemma} \label{chi} 
We have
\bEq \label{EspressioneChi}
\sum_{|I|=0}^{r-1}\text{d}_i\text{d}_I (\chi^{i_1\dots i_siI}-\chi^{[i_1\dots i_si]I} )\wedge ds_{i_1\dots i_s} \equiv 0 \,.
\eEq
\end{lemma}

\begin{proof}
Let us consider the terms
\beq
\text{d}_i \chi^{i_1\dots i_siI} \wedge ds_{i_1\dots i_s}
\eeq
here a sum on the index $i$ is understood, let us write it explicitly in order to manipulate it:
\beq
\text{d}_1 \chi^{i_1\dots i_s 1I} \wedge ds_{i_1\dots i_s} + \text{d}_2 \chi^{i_1\dots i_s2I} \wedge ds_{i_1\dots i_s} + \text{d}_3 \chi^{i_1\dots i_s 3I}  \wedge ds_{i_1\dots i_s}+ \dots
\eeq
Let ${i}_p\neq i$,  $1\leq p\leq s$, we have
\beq
ds_{{i}_1\dots {i}_s}=dx^i\wedge ds_{{i}_1\dots {i}_s i}
\eeq
which holds for each $i\neq {i}_1,\dots, {i}_s $ and is intended {\em without summation over $i$}.
Let us substitute for each convenient index in the summation. We get
\beq
\text{d}_1 \chi^{i_1\dots i_s 1I} \wedge dx^1\wedge ds_{{i}_1\dots {i}_s 1}+ \dots+\text{d}_i \chi^{i_1\dots i_s iI}  \wedge dx^i\wedge ds_{{i}_1\dots {i}_s i}  +\dots
\eeq
Now, by rearranging the summation (see in detail the procedure in the proof of next Proposition)
note that in $ds_{{i}_1\dots {i}_si}$ the indices ${i}_1\dots {i}_s i$ are antisymmetric for each $i$ therefore necessarily each of the coefficients must satisfy $\chi^{i_1\dots i_s 1I}= \chi^{[i_1\dots i_s 1]I}$, $\chi^{i_1\dots i_s 2I}= \chi^{[i_1\dots i_s 2]I}$, \dots, therefore for each $i\neq {i}_1,\dots, {i}_s $
$\chi^{i_1\dots i_s iI}= \chi^{[i_1\dots i_s i]I}$.
Thus if ${i}_p\neq i$,  $1\leq p\leq s$,
\beq
\sum_{|I|=0}^{r-1}\text{d}_i\text{d}_I (\chi^{i_1\dots i_siI}-\chi^{[i_1\dots i_si]I} )\wedge ds_{i_1\dots i_s} = \\
\sum_{|I|=0}^{r-1}\text{d}_j\text{d}_I (\chi^{[i_1\dots i_si]I}-\chi^{[i_1\dots i_si]I} ) \wedge dx^j \wedge ds_{i_1\dots i_s i} \equiv 0 \,.
\eeq

Let us now discuss the cases  $i = i_p $,  for some $1\leq p\leq s$. In each of these cases we have obviously $\chi^{[i_1\dots  i i_si]I}=0$, being the index $i$ repeated twice.
On the other hand, taking into account the antisymmetry of $ds_{i_1\dots i i_s}$ in all its indices, and the fact that we have two sums on indices $i$'s which always must coincide, we can easily check that
\beq
 \sum_{|I|=0}^{r-1}\text{d}_i\text{d}_I \chi^{i_1\dots ( i i_s) iI}  \wedge ds_{i_1\dots i i_s} \equiv 0 \implies
 \sum_{|I|=0}^{r-1}\text{d}_i\text{d}_I \chi^{i_1\dots i (i_si)I}  \wedge ds_{i_1\dots i  i_s} \equiv 0 \,,
\eeq
This in turn implies that
\beq
 \sum_{|I|=0}^{r-1}\text{d}_i\text{d}_I \chi^{i_1\dots i i_s iI}  \wedge ds_{i_1\dots i i_s} =
  \sum_{|I|=0}^{r-1}\text{d}_i\text{d}_I \chi^{i_1\dots [ i i_s] iI} \wedge ds_{i_1\dots i i_s} \,,
  \eeq
but this last term is equal (up to a numerical coefficient) to  $\sum_{|I|=0}^{r-1}\text{d}_i\text{d}_I  \chi^{i_1\dots   i_s [ i i]I} $ $  \wedge $ $ ds_{i_1\dots i i_s} $ which is identically vanishing.

Indeed, let 
$d_i \Omega^{i_1\dots i_{s-2} i i_s i} \byd \sum_{|I|=0}^{r-1}\text{d}_i\text{d}_I \chi^{i_1\dots i i_si I}$. Since $d_i \Omega^{[i_1\dots i_{s-2} i i_s i]}=0$ then $d_i \Omega^{i_1\dots i_{s-2} i i_s i}$ is reduced according to Definition \ref{def:1.5}. 
On the other hand we also have $d_i \Omega^{i_1\dots i_{s-2} i (i_s i)}=0$,
thus (as a consequence of \cite{FF:Natural}, Lemma $6.2.43$ p. $164$) we recover the fact  that
$d_i \Omega^{i_1\dots i_{s-2} i i_s i} = d_i \Omega^{i_1\dots i_{s-2} i i_s i_{s-1}}  \equiv 0$.
We shall specify this result for  $k=1$ in the next Section (see Example \ref{2}).
\end{proof}

Let us now explicate the second summand of equation \eqref{eq:dec1} which will play a fundamental role for the application to Lepage equivalents in Section \ref{Lepage equivalents}.
\begin{prop} \label{prop:div}
Let $\rho\in\Omega^r_{n-s+k}W$ be a $(n-s)$-horizontal $k$-contact $(n-s+k)$-form defined on the $r$-order jet prolongation $W^r$ of $W$. 
We have 
\bEq \label{eq:dh1}
& &  \text{d}_i\sum_{|I|=0}^{r-1}\text{d}_I\chi^{[i_1\dots i_si]I}\wedge ds_{i_1\dots i_s}=\\ 
& &  d_H\bigg(\sum_{|I|=0}^{r-1}(-1)^k\frac{1}{(s+1)}\text{d}_I\chi^{[i_1\dots i_si]I}\wedge ds_{i_1\dots i_si}\bigg) =d_H \mathscr{R}(\rho) \,, \nonumber
\eEq
 where $\mathscr{R}(\rho)$ is a local $(n-s-1)$-horizontal $k$-contact $(n-s-1+k)$-form.
\end{prop}
\begin{proof}
First we rewrite expression \eqref{eq:dh1} using a summation over \textit{ordered} indices $\tilde{i}_1\le\dots\le\tilde{i}_s$ instead of $i_1\dots i_s$:
\beq
\text{d}_i\sum_{|I|=0}^{r-1}\text{d}_I\chi^{[i_1\dots i_si]I}\wedge ds_{i_1\dots i_s}=\text{d}_i\sum_{|I|=0}^{r-1}s!\,\text{d}_I\chi^{[\tilde{i}_1\dots\tilde{i}_si]I}\wedge ds_{\tilde{i}_1\dots\tilde{i}_s}.
\eeq
Then we expand this sum, using Einstein's convention for the summation over multi-indices $I$:
\beq
& & \text{d}_i\sum_{|I|=0}^{r-1}s!\, \text{d}_I\chi^{[\tilde{i}_1\dots\tilde{i}_si]I}\wedge ds_{\tilde{i}_1\dots\tilde{i}_s}= \\
& & =s!\Big( \text{d}_1\text{d}_I\chi^{[\tilde{i}_1\dots\tilde{i}_s1]I}\wedge ds_{\tilde{i}_1\dots\tilde{i}_s}+
\text{d}_2\text{d}_I\chi^{[\tilde{i}_1\dots\tilde{i}_s2]I}\wedge ds_{\tilde{i}_1\dots\tilde{i}_s}+\dots \\
& &  \dots +\text{d}_n\text{d}_I\chi^{[\tilde{i}_1\dots\tilde{i}_sn]I}\wedge ds_{\tilde{i}_1\dots\tilde{i}_s}\Big).
\eeq
Let now again write
\beq
ds_{\tilde{i}_1\dots\tilde{i}_s}=dx^i\wedge ds_{\tilde{i}_1\dots\tilde{i}_si}
\eeq
without summation on the index $i$ and with $i\neq\tilde{i}_1,\dots,\tilde{i}_s$. Then we can proceed as follows:
\bEq
\label{eq:passo2}
& & \text{d}_i\sum_{|I|=0}^{r-1}s!\,\text{d}_I\chi^{[\tilde{i}_1\dots\tilde{i}_si]I}\wedge ds_{\tilde{i}_1\dots\tilde{i}_s}= \\
& & = s!\Big(\text{d}_1\text{d}_I\chi^{[\tilde{i}_1\dots\tilde{i}_s1]I}\wedge dx^1\wedge ds_{\tilde{i}_1\dots\tilde{i}_s1}+
\text{d}_2\text{d}_I\chi^{[\tilde{i}_1\dots\tilde{i}_s2]I}\wedge dx^2\wedge ds_{\tilde{i}_1\dots\tilde{i}_s2}+\dots     \nonumber  \\
& & \dots +\text{d}_n\text{d}_I\chi^{[\tilde{i}_1\dots\tilde{i}_sn]I}\wedge dx^n\wedge ds_{\tilde{i}_1\dots\tilde{i}_sn}\Big) \,. \nonumber 
\eEq
For each term of the last equation, we have the same indices in $\chi^{[\tilde{i}_1\dots\tilde{i}_si]I}$ and in $ds_{\tilde{i}_1\dots\tilde{i}_si}$, though without summation on $i$. 

Our goal is to take each term $\text{d}_i\text{d}_I\chi^{[\tilde{i}_1\dots\tilde{i}_si]I}\wedge dx^i\wedge ds_{\tilde{i}_1\dots\tilde{i}_si}$ without summation on any index and find a way to write it in the form
\beq
\text{d}_l\text{d}_I\chi^{[\tilde{i}_1\dots\tilde{i}_si]I}\wedge dx^l\wedge ds_{\tilde{i}_1\dots\tilde{i}_si}
\eeq
with summation on the index $l$. In order to make explicit the reasoning, let us consider the term with $i=1$
\beq
\text{d}_1\text{d}_I\chi^{[\tilde{i}_1\dots\tilde{i}_s1]I}\wedge dx^1\wedge ds_{\tilde{i}_1\dots\tilde{i}_s1}
\eeq
again without summation over any index. We have obviously that:
\beq
\text{d}_1\text{d}_I\chi^{[\tilde{i}_1\dots\tilde{i}_s1]I}\wedge dx^1\wedge ds_{\tilde{i}_1\dots\tilde{i}_s1}=\sum_{l\neq\tilde{i}_1,\dots,\tilde{i}_s}\text{d}_l\text{d}_I\chi^{[\tilde{i}_1\dots\tilde{i}_s1]I}\wedge dx^l\wedge ds_{\tilde{i}_1\dots\tilde{i}_s1}.
\eeq
The remaining terms when $l=\tilde{i}_1,\dots,l=\tilde{i}_s$ can be found among the other summands of the right hand side of equation \eqref{eq:passo2}. As an example, consider the case when $l=\tilde{i}_1$: among the terms
\beq
\text{d}_{\tilde{i}_1}\text{d}_I\chi^{[\tilde{j}_1\dots\tilde{j}_s\tilde{i}_1]I}\wedge dx^{\tilde{i}_1}\wedge ds_{\tilde{j}_1\dots\tilde{j}_s\tilde{i}_1}
\eeq
(with summation on the ordered indices $\tilde{j}_1\le\dots\le\tilde{j}_s$) there will certainly be a summand
of the form
\beq
\text{d}_{\tilde{i}_1}\text{d}_I\chi^{[1\tilde{i}_2\dots\tilde{i}_s\tilde{i}_1]I}\wedge dx^{\tilde{i}_1}\wedge ds_{1\tilde{i}_2\dots\tilde{i}_s\tilde{i}_1}
\eeq
(this time without summation on any index) which can be recast as
\beq
\text{d}_{\tilde{i}_1}\text{d}_I\chi^{[\tilde{i}_1\dots\tilde{i}_s1]I}\wedge dx^{\tilde{i}_1}\wedge ds_{\tilde{i}_1\dots\tilde{i}_s1}.
\eeq
This is exactly the term we were searching for. Proceeding in the same way for $l=\tilde{i}_2,\dots,l=\tilde{i}_s$, we can finally obtain the expression
\beq
\text{d}_l\text{d}_I\chi^{[\tilde{i}_1\dots\tilde{i}_s1]I}\wedge dx^l\wedge ds_{\tilde{i}_1\dots\tilde{i}_s1}
\eeq
(with summation on $l$) as wanted. Moreover, we remark that the total number of summands of equation \eqref{eq:passo2} is $N=n\binom{n-1}{s}$, while the number of ordered strings of $(s+1)$ indices is $N'=\binom{n}{s+1}$, and $N=(s+1)N'$. This means that for every single term of equation \eqref{eq:passo2} $\text{d}_i\text{d}_I\chi^{[\tilde{i}_1\dots\tilde{i}_si]I}\wedge dx^i\wedge ds_{\tilde{i}_1\dots\tilde{i}_si}$ (without summation on any index), there are other $s$ terms with the same indices, just in a different order.

Therefore, it is not difficult to see that
\beq
\text{d}_i\sum_{|I|=0}^{r-1}s!\,\text{d}_I\chi^{[\tilde{i}_1\dots\tilde{i}_si]I}\wedge ds_{\tilde{i}_1\dots\tilde{i}_s}=\text{d}_l\sum_{|I|=0}^{r-1}s!\,\text{d}_I\chi^{[\tilde{i}_1\dots\tilde{i}_s\tilde{i}_{s+1}]I}\wedge dx^l\wedge ds_{\tilde{i}_1\dots\tilde{i}_s\tilde{i}_{s+1}}
\eeq
with summation on ordered indices $\tilde{i}_1\le\dots\le\tilde{i}_s\le\tilde{i}_{s+1}$. Passing to a summation on non-ordered indices $i_1\dots i_{s+1}$ and using commutation properties of wedge products, we can finally write:
\beq
& & \text{d}_i \sum_{|I|=0}^{r-1}\text{d}_I\chi^{[i_1\dots i_si]I}\wedge ds_{i_1\dots i_s}= \\
& & = \text{d}_l\sum_{|I|=0}^{r-1}\frac{1}{(s+1)}\,\text{d}_I\chi^{[i_1\dots i_{s+1}]I}\wedge dx^l\wedge ds_{i_1\dots i_{s+1}}= \\
& & =(-1)^{n-s-1+k}\text{d}_l\bigg(\sum_{|I|=0}^{r-1}(-1)^k\frac{1}{(s+1)}\,\text{d}_I\chi^{[i_1\dots i_{s+1}]I}\wedge ds_{i_1\dots i_{s+1}}\bigg)\wedge dx^l= \\
& & =d_H\bigg(\sum_{|I|=0}^{r-1}(-1)^k\frac{1}{(s+1)}\,\text{d}_I\chi^{[i_1\dots i_{s+1}]I}\wedge ds_{i_1\dots i_{s+1}}\bigg) \byd d_H \mathscr{R}(\rho).
\eeq
\end{proof}

\bDf
Let $\rho\in\Omega^r_{n-s+k}W$ be a $(n-s)$-horizontal $k$-contact $(n-s+k)$-form and let $
p_k\rho=\sum_{|J|=0}^r\omega^\sigma_J\wedge\eta^J_\sigma\,\in\Omega^{r+1}_{n-s+k}W
$.

We define the {\em Residual operator associated with $\rho$} as the operator $\mathscr{R}$ locally characterized by the above Proposition where the forms $\chi$'s are defined 
in terms of the forms $\eta$'s by equations \eqref{KrMu lemma} and \eqref{eq:sost}. 
\eDf

\subsection{Comparison in the case $k=1$}

In order to compare the results above with the canonical splitting of variational morphisms, let us restrict to the case $k=1$.

\begin{prop} \label{prop:splitlike} 
Let $p_1\rho=\sum_{|L|=0}^r\omega^\sigma_L\wedge\eta^L_\sigma\,\in\Omega^{r+1}_{n-s+1}W$, with 
\beq
\eta^L_\sigma=A^{i_1\dots i_sL}_\sigma \wedge ds_{i_1\dots i_s}
\eeq
where  the coefficients $A^{i_1\dots i_s L}_\sigma$ are defined on $J^{r+1}W$.
We have a ``canonical splitting''-like decomposition
\bEq
\label{eq:splittinglike}
& p_1\rho= \sum_{k=0}^{r} (-1)^k d_{a_k \dots a_3 a_2 a_1}  A^{[i_1\dots i_s a_1]a_2 a_3 \dots a_k}_\sigma \omega^\sigma\wedge ds_{i_1\dots i_s} +\\  
& + d_H [\frac{1}{(s+1)}\sum_{|L|=0}^{r-1} - \hat{t}^{i_1\dots i_siL}_\sigma\omega^\sigma_L\wedge ds_{i_1\dots i_si} ]  \nonumber 
\eEq
where the coefficients $\hat{t}$'s are defined iteratively by:
\bEq
\label{eq:itert}
& & \hat{t}^{i_1\dots i_sil_1\dots l_{r-1}}_\sigma=A^{[i_1\dots i_si]l_1\dots l_{r-1}}_\sigma  \nonumber \\
& & \hat{t}^{i_1\dots i_sil_1\dots l_{r-2}}_\sigma=A^{[i_1\dots i_si]l_1\dots l_{r-2}}_\sigma-d_k\hat{t}^{i_1\dots i_sil_1\dots l_{r-2}k}_\sigma \\
& & \dots  \nonumber \\
& & \hat{t}^{i_1\dots i_si}_\sigma=A^{[i_1\dots i_si]}_\sigma-d_k\hat{t}^{i_1\dots i_sik}_\sigma \,. \nonumber 
\eEq
\end{prop}

\begin{proof}
We have
\beq
\xi^I_\sigma= (\sum_{|J|=0}^{r-|I|}(-1)^{|J|}\binom{|J|+|I|}{|J|}d_JA^{i_1\dots i_sIJ}_\sigma )ds_{i_1\dots i_s}
\eeq
and consequently, in the case $k=1$, we put
\bEq \label{eq:chi}
\chi^{i_1\dots i_sI}=\sum_{|J|=0}^{r-|I|}(-1)^{|J|}\binom{|J|+|I|}{|J|}d_JA^{i_1\dots i_sIJ}_\sigma\omega^\sigma \,,
\eEq
therefore Proposition \ref{prop:div} specializes as
\bEq
\label{eq:div}
& & \text{d}_i\sum_{|I|=0}^{r-1}\text{d}_I\chi^{[i_1\dots i_si]I}\wedge ds_{i_1\dots i_s} = \\
& & = d_H [\sum_{|I|=0}^{r-1}\frac{-1}{(s+1)}\text{d}_I\chi^{[i_1\dots i_si]I}\wedge ds_{i_1\dots i_si} ]= \nonumber \\
& & =d_H [\frac{-1}{(s+1)}\sum_{|I|=0}^{r-1}\text{d}_I (\sum_{|J|=0}^{r-|iI|}(-1)^{|J|}\binom{|J|+|iI|}{|J|}d_JA^{[i_1\dots i_si]IJ}_\sigma\omega^\sigma)\wedge ds_{i_1\dots i_si} ]. \nonumber 
\eEq
The explicit expression of this term can be computed as described in \textit{Step 2} of the proof of Proposition \ref{prop:Volume}. In particular, if we develop the total derivatives $\text{d}_I$ inside the sum in equation \eqref{eq:div} and collect the coefficients of the contact forms of the same order, we obtain the expression
\bEq
\label{eq:t}
& & d_H\bigg[\frac{-1}{(s+1)}\sum_{|I|=0}^{r-1}\text{d}_I\Big(\sum_{|J|=0}^{r-|iI|}(-1)^{|J|}\binom{|J|+|iI|}{|J|}d_JA^{[i_1\dots i_si]IJ}_\sigma\omega^\sigma\Big)\wedge ds_{i_1\dots i_si}\bigg]= \nonumber \\
& &  =d_H [-\frac{1}{(s+1)}\sum_{|L|=0}^{r-1}\hat{t}^{i_1\dots i_siL}_\sigma\omega^\sigma_L\wedge ds_{i_1\dots i_si} ] \,.
\eEq
Because of its role in the extension of the interior Euler operator to the cases of lower degree, we now proceed in particular  to 
determine  the coefficient of $\omega^\sigma\wedge ds_{i_1\dots i_s}$. 
The first term we need to consider is obviously  
\beq
A^{i_1\dots i_s}_\sigma \omega^\sigma \wedge ds_{i_1\dots i_s} \,.
\eeq
Put $B^{i_1\dots i_siIJ}_\sigma=A^{i_1\dots i_siIJ}_\sigma-A^{[i_1\dots i_si]IJ}_\sigma$.
The other contributions then,  of course, come from
\bEq \label{omegasigma}
\sum_{|I|=0}^{r-1}\text{d}_i\text{d}_I\bigg(\sum_{|J|=0}^{r-|iI|}(-1)^{|J|}\binom{|J|+|iI|}{|J|}d_JB^{i_1\dots i_siIJ}_\sigma \bigg) \omega^\sigma\wedge ds_{i_1\dots i_s} \,.
\eEq 
In other words, renaming $iI\to I$, we compute the expression:
\beq
\sum_{|J|=0}^r(-1)^{|J|}d_JA^{i_1\dots i_sJ}_\sigma+\sum_{|I|=1}^r\sum_{|J|=0}^{r-|I|}(-1)^{|J|}\binom{|J|+|I|}{|J|}d_Id_JB^{i_1\dots i_sIJ}_\sigma.
\eeq
We collect the terms according to the length $| L| $ of the multi-index $J$ in the first sum and $IJ$ in the second sum:
\begin{description}
\item[$ | L |=0$] Since the second double sum   starts with $|I|=1$, we have only the first term $A^{i_1\dots i_s}_\sigma$ from the first sum.
\item[$  | L |=1$] From the first sum we get the term $-d_iA^{i_1\dots i_si}_\sigma$, while from the second sum, considering $|I|=1$ and $|J|=0$, we get $d_iB^{i_1\dots i_si}_\sigma$. Summing up these contribution we obtain:
\beq
-d_iA^{i_1\dots i_si}_\sigma+d_iB^{i_1\dots i_si}_\sigma=-d_iA^{[i_1\dots i_si]}_\sigma.
\eeq
\item[$ | L |=2$] From the first sum we get the term $d_{ab}A^{i_1\dots i_sab}_\sigma$, while in the second sum we can choose $|I|=2$ and $|J|=0$, getting the term $d_{ab}B^{i_1\dots i_sab}_\sigma$, or $|I|=1$ and $|J|=1$, getting the term $-2d_{ab}B^{i_1\dots i_sab}_\sigma$. Summing all together, we obtain:
\beq
d_{ab}A^{i_1\dots i_sab}_\sigma-d_{ab}B^{i_1\dots i_sab}_\sigma=d_ad_bA^{[i_1\dots i_sa]b}_\sigma.
\eeq
\end{description}
It is not difficult to see that, for general length $| L|$, we obtain the term:
\beq
& & (-1)^{| L |} d_{a_1\dots a_{ | L |}}A^{i_1\dots i_sa_1\dots a_{ | L |}}_\sigma+\sum_{k=0}^{ | L| -1}(-1)^k\binom{| L |}{k}d_{a_1\dots a_{ | L |}}B^{i_1\dots i_sa_1\dots a_{ | L |}}_\sigma= \\
& & =(-1)^{| L |}  \Big(d_{a_1\dots a_{ | L |}}A^{i_1\dots i_sa_1\dots a_{ | L |}}_\sigma-d_{a_1\dots a_{ | L |}}B^{i_1\dots i_sa_1\dots a_{ | L |}}_\sigma\Big)= \\
& & =(-1)^{| L |}  d_{a_1\dots a_{ | L |}}A^{[i_1\dots i_sa_1]a_2\dots a_{ | L |}}_\sigma.
\eeq
where we used
$\sum_{k=0}^{|L | -1}(-1)^k\binom{|L|}{k}=-(-1)^{| L |} $. Thence finally, the coefficient of the form $\omega^\sigma\wedge ds_{i_1\dots i_s}$ is given by:
\bEq \label{Interior lower}
 A^{i_1\dots i_s}_\sigma-d_{a_1}A^{[i_1\dots i_sa_1]}_\sigma+d_{a_1}d_{a_2}A^{[i_1\dots i_sa_1]a_2}_\sigma+\dots \\ 
 \dots +(-1)^rd_{a_1}\dots d_{a_r}A^{[i_1\dots i_sa_1]a_2\dots a_r}_\sigma \,. \nonumber
\eEq

We now study the terms generated by $\omega^\sigma_L\wedge ds_{i_1\dots i_s}$. 
We explicate the total derivatives of forms in the following expression
\bEq
\label{eq:calcolisplitting}
& & \sum_{|I|=0}^{r-1}\text{d}_i\text{d}_I\Big(\chi^{i_1\dots i_siI}-\chi^{[i_1\dots i_si]I}\Big)\wedge ds_{i_1\dots i_s}=  \\
& & =\sum_{|I|=0}^{r-1}\text{d}_i\text{d}_I\bigg(\sum_{|J|=0}^{r-|iI|}(-1)^{|J|}\binom{|J|+|iI|}{|J|}d_JB^{i_1\dots i_siIJ}_\sigma \omega^\sigma  \bigg) \wedge ds_{i_1\dots i_s} \, \nonumber 
\eEq
and get the coefficients of these terms according to the length of $L$.
 
\begin{description}
\item[$|L|=r$] The only contribution to this coefficient comes from setting $|I|=r-1$ in equation \eqref{eq:calcolisplitting} and applying all the total derivatives to $\omega^\sigma$. Renaming the multi-index $L=lL'$ (hence with $|L'|=|L|-1$), we obtain:
\beq
B^{i_1\dots i_slL'}_\sigma=A^{i_1\dots i_slL'}_\sigma-A^{[i_1\dots i_sl]L'}_\sigma.
\eeq
\item[$|L|=r-1$] One first contribution comes from setting $|I|=r-2$ and applying all the total derivatives to $\omega^\sigma$, getting:
\beq
B^{i_1\dots i_sL}_\sigma-rd_iB^{i_1\dots i_sLi}_\sigma.
\eeq
Another contribution comes from setting $|I|=r-1$ and applying one derivative to $B$ and the others to $\omega^\sigma$, getting:
\beq
d_iB^{i_1\dots i_siL}_\sigma+\binom{r-1}{1}d_iB^{i_1\dots i_sLi}_\sigma.
\eeq
Summing up these contributions, and renaming $L=lL'$, we obtain:
\beq
& & B^{i_1\dots i_slL'}_\sigma-d_iB^{i_1\dots i_slL'i}_\sigma+d_iB^{i_1\dots i_silL'}_\sigma= \\
& &  =A^{i_1\dots i_slL'}_\sigma-A^{[i_1\dots i_sl]L'}_\sigma-d_iA^{i_1\dots i_slL'i}_\sigma+d_iA^{[i_1\dots i_sl]L'i}_\sigma+ \\
& & +d_iA^{i_1\dots i_silL'}_\sigma-d_iA^{[i_1\dots i_si]lL'}_\sigma= \\
& & =A^{i_1\dots i_slL'}_\sigma-A^{[i_1\dots i_sl]L'}_\sigma+d_iA^{[i_1\dots i_sl]L'i}_\sigma-d_iA^{[i_1\dots i_si]lL'}_\sigma.
\eeq
\item[$|L|=r-2$] One contribution comes from setting $|I|=r-3$ and applying all the derivatives to $\omega^\sigma$. Another contribution comes from setting $|I|=r-2$ and applying only $r-2$ derivatives to $\omega^\sigma$. One last contribution comes from setting $|I|=r-1$ and applying only $r-2$ derivatives to $\omega^\sigma$.
Summing up, the coefficient for $|L|=r-2$ is:
\beq 
&  A^{i_1\dots i_slL'}_\sigma-A^{[i_1\dots i_sl]L'}_\sigma+d_iA^{[i_1\dots i_sl]L'i}_\sigma -d_iA^{[i_1\dots i_si]lL'}_\sigma \\
&  \quad - d_{ai}A^{[i_1\dots i_sl]L'ai}_\sigma  +d_{ia}A^{[i_1\dots i_si]alL'}_\sigma \,. \nonumber 
\eeq 

\item[$|L|=r-g$, $0\leq g \leq r$]  in general we have
\beq 
\sum_{m=0}^{g}(-1)^{m}  d_{j_m\dots j_1}(A^{[i_1\dots i_sl]L'  j_{m} \dots j_1}_\sigma  
- A^{[i_1\dots i_sj_m] j_{m-1} \dots j_1 lL' }_\sigma) \omega^\sigma_{lL'} \wedge ds_{i_1\dots i_s}  \,.
\eeq 
\end{description}
Note that we recover expressions \eqref{Interior lower} for $|L|=|lL'|=0$.
 
Integrating by parts we obtain
\beq
\sum_{m=0}^{g}(-1)^{m} d_l (d_{j_m\dots j_1}(A^{[i_1\dots i_sl]L'  j_{m} \dots j_1}_\sigma  - A^{[i_1\dots i_sj_m] j_{m-1} \dots  j_1 lL' }_\sigma) \omega^\sigma_{L'} ) \wedge ds_{i_1\dots i_s}  +
\\
- \sum_{m=0}^{g}(-1)^{m} d_l (d_{j_m\dots  j_1}(A^{[i_1\dots i_sl]L'  j_{m} \dots  j_1}_\sigma  - A^{[i_1\dots i_sj_m] j_{m-1} \dots   j_1 lL' }_\sigma) )\omega^\sigma_{L'} \wedge ds_{i_1\dots i_s} 
\eeq
where the last piece (apart the case $g=0$, contributing the coefficient of $\omega^\sigma_{L'} \wedge ds_{i_1\dots i_s}$) is zero because of the symmetry in the indices $l$ and $j_m$. 

For $g\neq 0$, and any choice of  $l= j_p$  $1\leq p\leq r $, taking into account the symmetry of the indices $j_m\dots j_p \dots j_3j_2j_1$, 
the two summands of the first piece become 
\beq
d_{j_p} d_{j_m\dots j_p \dots j_3j_2j_1}(A^{[i_1\dots i_s j_p ]L'  j_{m} \dots j_2j_1}_\sigma  - A^{[i_1\dots i_s j_m] j_{m-1} \dots j_2 j_1 j_p L' }_\sigma) \equiv 0 \,.
\eeq
For $g=0$ and  $l\neq i_p$
\beq
 d_l ( (A^{[i_1\dots i_sl]L' }_\sigma  - A^{[i_1\dots i_s ]  l L'}_\sigma) \omega^\sigma_{L'} ) \wedge ds_{i_1\dots i_s} 
 =  \\ d_l ((A^{[i_1\dots [i_s l ]]L' }_\sigma  - A^{[i_1\dots [i_s ]  l ]L' }_\sigma) \omega^\sigma_{L'} ) \wedge dx^l \wedge ds_{i_1\dots i_s l} 
 = 0 \,.
\eeq
For the case $l = i_p$, for some $p$, we refer to the discussion in the general case $g\geq1$, Lemma \ref{chi}.
\end{proof}

\bDf
Expression  \eqref{Interior lower} defines a {\em local interior Euler operator for lower degree $1$-contact forms}
\bEq \label{lowinterior}
\mathfrak{I}(\rho) 
= \omega^\sigma\wedge (A^{i_1\dots i_s}_\sigma-d_{a_1}A^{[i_1\dots i_sa_1]}_\sigma+d_{a_1}d_{a_2}A^{[i_1\dots i_sa_1]a_2}_\sigma+\dots \nonumber \\ 
 \dots +(-1)^rd_{a_1}\dots d_{a_r}A^{[i_1\dots i_sa_1]a_2\dots a_r}_\sigma  )ds_{i_1\dots i_s} \,.
\eEq
It is an expression analogous to \eqref{topInterior} and the study of its uniqueness and  globality properties will be done elsewhere.
\eDf

\bRm \label{antisymm}
For the cases $s=1$ and $k\geq1$, the splitting 
$
p_1\rho=\omega^\sigma\wedge\xi_\sigma+d_H(\sum_{|I|=0}^{r-1}-\frac{1}{(s+1)}\text{d}_I\chi^{i_1\dots i_siI}\wedge ds_{i_1\dots i_si})
$
first appeared in \cite{Palese:VariationalSeq}, however the explicit local expressions for $\xi_\sigma$ are obtained here for the first time (for $k=1$). Furthermore, a Residual operator is here locally characterized {\em for any $s\geq 1$}. 
\eRm

\bRm
Let us consider 
$p_1\rho\in\Omega^{r+1}_{n-s+1}W$ as a variational morphism
\beq
\mathbb{V}_\rho\colon J^{r+1}W\longrightarrow (J^rV(W))^\ast\otimes A_{n-s}(U)
\eeq
with
$
<\mathbb{V}_\rho|J^r\Xi> $ $:=$ $J^r\Xi\lrcorner\,p_1\rho $ $=$ $
(A^{i_1\dots i_s}_\sigma\Xi^\sigma+\dots +A^{i_1\dots i_sj_1\dots j_r}_\sigma\Xi^\sigma_{j_1\dots j_r})\wedge ds_{i_1\dots i_s}$.
On the other hand, according to Proposition \ref{prop:splitlike}, with  
$p_1\rho\in\Omega^{r+1}_{n-s+1}W$ we also associate two variational morphisms
\beq
& &\mathbb{E}'\colon J^{2r+1}W\longrightarrow (J^{r}V(W))^\ast\otimes A_{n-s}(U) \\
& &\mathbb{T}'\colon J^{2r}W\longrightarrow (J^{r-1}V(W))^\ast\otimes A_{n-s-1}(U),
\eeq
such that 
$
<\mathbb{V}_\rho|J^r\Xi>=<\mathbb{E}'|J^r\Xi> $ $+$ $\text{Div}(<\mathbb{T}'|J^{r-1}\Xi>)$ and in particular we have the correspondences:
$
<\mathbb{E}'|J^r\Xi> $ $=$ $J^r\Xi\lrcorner (\omega^\sigma\wedge\xi_\sigma
)$, 
$\text{Div}(<\mathbb{T}'|J^{r-1}\Xi>)$ $=$ $J^r\Xi\lrcorner\,d_H(\sum_{|I|=0}^{r-1}-\frac{1}{(s+1)}\text{d}_I\chi^{[i_1\dots i_si]I}\wedge ds_{i_1\dots i_si})$.

By equations \eqref{eq:t} and \eqref{eq:itert} 
\beq
<\mathbb{T}'|J^{r-1}\Xi>=(\hat{t}^{i_1\dots i_{s+1}}_\sigma\Xi^\sigma + \dots +\hat{t}^{i_1\dots i_{s+1}j_1\dots j_{r-1}}_\sigma\Xi^\sigma_{j_1\dots j_{r-1}})\wedge ds_{i_1\dots i_si_{s+1}}. \,,
\eeq
while 
\beq
<\mathbb{E}'|J^r\Xi>= (\hat{e}^{i_1\dots i_s}_\sigma\Xi^\sigma+\hat{e}^{i_1\dots i_sj_1}_\sigma\Xi^\sigma_{j_1}+\dots +\hat{e}^{i_1\dots i_sj_1\dots j_r}_\sigma\Xi^\sigma_{j_1\dots j_r})\wedge ds_{i_1\dots i_s}
\eeq
and the general form of the coefficients $\hat{e}$ are given by Proposition \ref{prop:splitlike}; they can be compared with the  canonical splitting  morphisms according to \cite{FF:Natural}. 
\eRm

In view of possible applications to \textit{momentum morphisms}, \ie $1$-contact $(n-1)$-forms (related to the Poincar\'e-Cartan morphism), we make now some explicit examples for the cases $r=1,2$.

\bEx \label{1}
Let $r=1$, $0\le s<n$ and let $\rho\in\Omega^1_{n-s+1}W$ be a $1$-contact $(n-s)$-horizontal form on $W^1$. Then the two splitting formulae:
\beq
<\mathbb{V}_\rho|J^1\Xi>=<\mathbb{E}|J^1\Xi>+\text{Div}(<\mathbb{T}|\Xi>)=  <\mathbb{E}'|J^1\Xi>+\text{Div}(<\mathbb{T}'|\Xi>)
\eeq
give the same result. In other words:
\beq
\mathbb{E}'=\mathbb{E} \,, \qquad
\mathbb{T}'=\mathbb{T}.
\eeq
\medskip
Indeed, we observe that locally we have
\beq
p_1\rho=\big(A^{i_1\dots i_s}_\sigma\omega^\sigma+A^{i_1\dots i_sj}_\sigma\omega^\sigma_j\big)\wedge ds_{i_1\dots i_s}\in\Omega^2_{n-s+1}W.
\eeq
Thence, for any vertical vector field $\Xi\colon W\longrightarrow V(W)$, the variational morphism $\mathbb{V}_\rho\colon J^2W\longrightarrow\big(J^1V(W)\big)^\ast\otimes A_{n-s}(U)$ is defined by:
\beq
<\mathbb{V}_\rho|J^1\Xi>=\frac{1}{s!}\big(s!A^{i_1\dots i_s}_\sigma\Xi^\sigma+s!A^{i_1\dots i_sj}_\sigma\Xi^\sigma_j\big)\wedge ds_{i_1\dots i_s}.
\eeq
Applying to this expression of $\mathbb{V}_\rho$ the canonical splitting algorithm from the theory of variational morphisms (see \cite{FF:Natural}), we have the following volume term and boundary term, respectively, which we compare with our `canonical splitting'-like decomposition \eqref{eq:splittinglike}:
\beq
& &  \mathbb{E}=\frac{1}{s!}\big[s!\big(A^{i_1\dots i_s}_\sigma-d_kA^{[i_1\dots i_sk]}_\sigma\big)\omega^\sigma+s!\big(A^{i_1\dots i_sj}_\sigma-A^{[i_1\dots i_sj]}_\sigma\big)\omega^\sigma_j\big]\wedge ds_{i_1\dots i_s} \\
 & &  = \mathbb{E}' + (A^{i_1\dots i_sj}_\sigma-A^{[i_1\dots i_sj]}_\sigma\big)\omega^\sigma_j \wedge ds_{i_1\dots i_s}\,.
\eeq
However, according to Proposition \ref{prop:splitlike}, $(A^{i_1\dots i_sj}_\sigma-A^{[i_1\dots i_sj]}_\sigma\big)\omega^\sigma_j \wedge ds_{i_1\dots i_s} \equiv 0$. Furthermore, 
\beq
& & \mathbb{T}=\frac{1}{(s+1)!}\big[s!A^{[i_1\dots i_si]}_\sigma\omega^\sigma\big]\wedge ds_{i_1\dots i_si} =  \mathbb{T}' \,.
\eeq
\eEx

\bEx \label{2}
Let $r=2$ and $s=1$ and let $\rho\in\Omega^2_nW$ be a $1$-contact $(n-1)$-horizontal form on $W^2$. Then there exists a variational morphism
\beq
\alpha\colon J^{4}W\longrightarrow (J^{1}V(W))^\ast\otimes A_{n-2}(U)
\eeq
such that the decompositions
\beq
 <\mathbb{V}_\rho|J^2\Xi> = <\mathbb{E}|J^2\Xi>+\text{Div}(<\mathbb{T}|J^{1}\Xi> )
  = <\mathbb{E}'|J^2\Xi>+\text{Div}(<\mathbb{T}'|J^{1}\Xi>)
\eeq
are related by
\bEq
\label{eq:Lepage}
\mathbb{E}'=\mathbb{E}-\mathcal{D}\alpha \,, \quad \mathbb{T}'=\mathbb{T}+\alpha \,, 
\eEq
where
$\mathcal{D}\alpha\colon J^{5}W\longrightarrow (J^{2}V(W))^\ast\otimes A_{n-1}(U)$
is the unique variational morphism such that
$<\mathcal{D}\alpha|J^2\Xi>=\text{Div}(<\alpha|J^{1}\Xi>)$. 

In fact, let
$
p_1\rho=\big(A^i_\sigma\omega^\sigma+A^{ij_1}_\sigma\omega^\sigma_{j_1}+A^{ij_1j_2}_\sigma\omega^\sigma_{j_1j_2}\big)\wedge ds_i\in\Omega^3_n W$. 
Choosing the fibered connection whose coefficients vanish in the coordinate domain $W\subset Y$,  
\beq
& & <\mathbb{E}|J^2\Xi>=\Big[\Big(A^i_\sigma-d_aA^{[ia]}_\sigma-\frac{2}{3}d_bd_aA^{[ib]a}_\sigma\Big)\Xi^\sigma+\\
& & +\Big(A^{(ij_1)}_\sigma+\frac{2}{3}d_aA^{aij_1}_\sigma-\frac{2}{3}d_aA^{(ij_1)a}_\sigma\Big)\Xi^\sigma_{j_1}+A^{(ij_1j_2)}_\sigma\Xi^\sigma_{j_1j_2}\Big]\wedge ds_i 
\,.
 \nonumber
 \eeq
After some manipulations 
\bEq \label{eq:difference} 
& & \mathbb{E} = 
 \mathbb{E}' 
 +\big[ 
 (A^{ij_1}_\sigma-d_aA^{[ia]j_1}_\sigma-A^{[ij_1]}_\sigma+d_aA^{[ij_1]a}_\sigma)\omega^\sigma_{j_1} +  \\
 & &
 + (A^{ij_1j_2}_\sigma-A^{[ij_1]j_2}_\sigma)\omega^\sigma_{j_1j_2}\big]\wedge ds_i    +\mathcal{D}(\alpha) 
 \,. \nonumber
\eEq
where
\beq
\mathcal{D}(\alpha)=   - \big[ \frac{1}{3}d_bd_aA^{[ib]a}_\sigma\omega^\sigma+(\frac{1}{3}d_aA^{[ij_1]a}_\sigma+\frac{1}{3}d_aA^{[ia]j_1}_\sigma)\omega^\sigma_{j_1}
 +\frac{1}{3}A^{[ij_1]j_2}_\sigma\omega^\sigma_{j_1j_2} \big] \wedge ds_i \,, \nonumber
\eeq
and one can check that the second summand of \eqref{eq:difference} vanishes.

In order to compute $\mathbb{T}'$ from equation \eqref{eq:div}, let
$
\chi^{ij_1}=(A^{ij_1}_\sigma-2d_aA^{ij_1a}_\sigma)\omega^\sigma$, 
$\chi^{ij_1j_2}=A^{ij_1j_2}_\sigma\omega^\sigma $, then
\beq
& &  \mathbb{T}'=\sum_{|J|=0}^1\frac{1}{2}d_J\chi^{[i_1i_2]J}\wedge ds_{i_1i_2} =\frac{1}{2}\chi^{[i_1i_2]}\wedge ds_{i_1i_2}+\frac{1}{2}d_a\chi^{[i_1i_2]a}\wedge ds_{i_1i_2}= \\
& &  =\big[\frac{1}{2}(A^{[i_1i_2]}_\sigma-d_aA^{[i_1i_2]a}_\sigma)\omega^\sigma+\frac{1}{2}A^{[i_1i_2]a}_\sigma\omega^\sigma_a\big]\wedge ds_{i_1i_2} \,.
\eeq
Comparing with
\beq
<\mathbb{T}|J^1\Xi>=\frac{1}{2}\Big[\Big(A^{[i_1i_2]}_\sigma-\frac{2}{3}d_aA^{[i_1i_2]a}_\sigma\Big)\Xi^\sigma+\frac{4}{3}A^{[i_1i_2]j}_\sigma\Xi^\sigma_j\Big]\wedge ds_{i_1i_2} \,.
\eeq
 it is immediate to get:
\beq
\alpha:=\mathbb{T}'-\mathbb{T}=-\frac{1}{6}[d_aA^{[i_1i_2]a}_\sigma\omega^\sigma-A^{[i_1i_2]a}_\sigma\omega^\sigma_a ]\wedge ds_{i_1i_2} = - \frac{1}{6}d_a(A^{[i_1i_2]a}_\sigma\omega^\sigma)\wedge ds_{i_1i_2}.
\eeq
We note that 
$\alpha$ can be put in the form of a local divergence and in agreement with Lemma \ref{chi} one can check that indeed $\mathcal{D}\alpha\equiv 0$.

Therefore we understand that the difference among the variational morphism splitting as stated in coordinates in \cite{FF:Natural} and our `canonical splitting"-like decomposition consists in passing from one splitting to the other one just by adding and substracting an identically vanishing divergence, an aspect which  deals with the local nature of our decomposition. 
\eEx

As we shall see, the relation \eqref{eq:Lepage} suggests a  sort of `canonical' construction of \textit{Lepage equivalents} other than the usual Poincar\'e-Cartan form (see also the discussion about this point in \cite{Palese:VariationalSeq}). 
The uniqueness and globality properties of the terms $\mathbb{E}'$ and $\mathbb{T}'$ remain to be deeper investigated; in particular, the difference $\mathbb{
T}'-\mathbb{T}$ is a closed form, and, if global, it defines a de Rham cohomology class. This feature deals with other topological aspects of Lagrangian field theories \cite{CaPaWi16,FePaWi11,FrPaWi13,Pa16,PaWi17}
which will be the subject of future research. 

\section{The Krupka--Betounes equivalent for first order theories and a first glance to the second order}  \label{Lepage equivalents}

Suppose we have a fibered manifold $\pi\colon Y\longrightarrow X$, with $\text{dim}X=n$ and $\text{dim}Y=n+m$. Consider a local chart $(x^i,y^\sigma)$, $1\le i\le n$, $1\le \sigma\le m$ on $Y$, then on $J^1Y$ we have local fibered coordinates $(x^i,y^\sigma,y^\sigma_j)$, with $1\le i,j\le n$ and $1\le \sigma\le m$.

Let us consider a Lagrangian $\lambda$, defined on $J^1Y$, locally expressed as $\lambda=\mathscr{L}(x^i,y^\sigma,y^\sigma_j)ds$. 
As already mentioned, in \cite{Palese:VariationalSeq} (Lemma 4.5, example 4.6) a Residual operator for lower degree forms was obtained by the first author for the case $s=1$. Olga Rossi \cite{OlgaRossi} conjectured that, making use of it, the Krupka--Betounes Lepage equivalent for first order theories could be obtained by the application of  the following recurrence formulae:
\beq
& & \rho_1=\lambda-p_1\mathcal{R}(d\lambda)=\theta_\lambda\quad\text{(Poincar\'e-Cartan form of the Lagrangian)} \\
& &  \rho_2=\rho_1 - p_2\mathscr{R}(d\rho_1) \\
& &  \rho_3=\rho_2-p_3\mathscr{R}(d\rho_2) \\
& & \dots \\
& &  \rho_n=\rho_{n-1}-p_n\mathscr{R}(d\rho_{n-1})
\eeq
with $\mathcal{R}$ and $\mathscr{R}$ the residual operator for top forms and lower degree forms, respectively.

In the following, in view of a generalization to the second order case, we apply the recurrence formulae by using the
characterization of the Residual operator for lower degree forms as given in the present paper by Proposition \ref{prop:div}.

We compute explicitly the forms $\rho_1,\dots,\rho_n$. First we have the well-known Poincar\'e-Cartan form of $\lambda$:
\beq
\rho_1=\theta_\lambda=\mathscr{L}ds+p^i_\sigma\omega^\sigma\wedge ds_i\qquad p^i_\sigma:=\frac{\partial\mathscr{L}}{\partial y^\sigma_i}.
\eeq
In order to compute $\rho_2$ we observe that the  Residual operator $\mathscr{R}$ does not change the order of contactness of its argument, i.e. $p_2\mathscr{R}(d\rho_1)=\mathscr{R}(p_2d\rho_1)$. Thus we can reduce ourselves to consider only the 2-contact component of the differential of $\rho_1$. In particular:
\beq
& & p_2d\rho_1 =\omega^\sigma\wedge (\partial_\sigma p^i_\nu\omega^\nu\wedge ds_i)+\omega^\sigma_j\wedge (\partial^j_\sigma p^i_\nu\omega^\nu\wedge ds_i)= \\
& & =\omega^\sigma\wedge [\partial_\sigma p^i_\nu\omega^\nu\wedge ds_i-d_j(\partial^j_\sigma p^i_\nu\omega^\nu\wedge ds_i)]+d_j[\omega^\sigma\wedge (\partial^j_\sigma p^i_\nu\omega^\nu\wedge ds_i)] \,.
\eeq
Specializing  Proposition \ref{prop:div}  we have:
\beq
\mathscr{R}(p_2d\rho_1)=\frac{1}{2}\partial^j_\sigma p^i_\nu\omega^\sigma\wedge\omega^\nu\wedge ds_{ij}.
\eeq
Thus
\beq
\rho_2=\rho_1 - p_2\mathscr{R}(d\rho_1)=\mathscr{L}ds+p^i_\sigma\omega^\sigma\wedge ds_i+\frac{1}{2}\partial^i_\sigma p^j_\nu\omega^\sigma\wedge\omega^\nu\wedge ds_{ij}.
\eeq
Now we compute $\rho_3=\rho_2-p_3\mathscr{R}(d\rho_2)$. Again we can restrict our attention to the term
\beq
&  p_3d\rho_2 = \frac{1}{2}(\partial_{\sigma_1}\partial^i_{\sigma_2}p^j_{\sigma_3}\omega^{\sigma_1}\wedge\omega^{\sigma_2}\wedge\omega^{\sigma_3}\wedge ds_{ij}+\partial^k_{\sigma_1}\partial^i_{\sigma_2}p^j_{\sigma_3}\omega^{\sigma_1}_k\wedge\omega^{\sigma_2}\wedge\omega^{\sigma_3}\wedge ds_{ij}) \\
&  = \omega^{\sigma_1}\wedge (\dots)+d_k[\omega^{\sigma_1}\wedge\frac{1}{2}(\partial^k_{\sigma_1}\partial^i_{\sigma_2}p^j_{\sigma_3}\omega^{\sigma_2}\wedge\omega^{\sigma_3}\wedge ds_{ij})].
\eeq
Thence
\beq
p_3\mathscr{R}(d\rho_2)=-\frac{1}{6}\partial^k_{\sigma_1}\partial^i_{\sigma_2}p^j_{\sigma_3}\omega^{\sigma_1}\wedge\omega^{\sigma_2}\wedge\omega^{\sigma_3}\wedge ds_{ijk}
\eeq
and finally
\beq
& & \rho_3=\rho_2-p_3\mathscr{R}(d\rho_2)= \mathscr{L}ds+p^i_\sigma\omega^\sigma\wedge ds_i+\frac{1}{2!}\partial^{i_1}_{\sigma_1}p^{i_2}_{\sigma_2}\omega^{\sigma_1}\wedge\omega^{\sigma_2}\wedge ds_{i_1i_2}+ \\
 & & +\frac{1}{3!}\partial^{i_1}_{\sigma_1}\partial^{i_2}_{\sigma_2}p^{i_3}_{\sigma_3}\omega^{\sigma_1}\wedge\omega^{\sigma_2}\wedge\omega^{\sigma_3}\wedge ds_{i_1i_2i_3}.
\eeq
Proceeding in this way it is straightforward to see that
\bEq\label{r=1}
\rho_n=\mathscr{L}ds+\sum_{q=1}^n\frac{1}{q!}\frac{\partial^q\mathscr{L}}{\partial y^{\sigma_1}_{i_1}\dots\partial y^{\sigma_q}_{i_q}}\omega^{\sigma_1}\wedge\dots\wedge\omega^{\sigma_q}\wedge ds_{i_1\dots i_q}
\eEq
which is known as the \textit{Krupka-Betounes equivalent} of the Lagrangian $\lambda$ \cite{Betounes84,
Kru73,Kru77,Perez19,Saunders10}.

\bEx
We use the notation adopted in Example 3.5 of article~\cite{Palese:VariationalSeq}. The $(1+1)$-dimensional free quantum particle is a first order theory for the fibered manifold $\pi\colon\mathbb{R}^2\times\mathbb{R}^2\rightarrow\mathbb{R}^2$, with coordinates $(t,x,v,w)$, described by the Lagrangian
\beq
\lambda = -\Big(\frac{\hbar^2}{4m}(v^2_x+w^2_x)+\frac{\hbar}{2}(vw_t-v_tw) \Big) \, dt\wedge dx = \mathscr{L} (j^1\pi)dt\wedge dx \,.
\eeq
In order to compute $\rho_1=\lambda-p_1\mathcal{R}(d\lambda)$, we first have to consider $p_1\mathcal{R}(d\lambda)=\mathcal{R}(p_1d\lambda)$. We have:
\beq
p_1d\lambda=\Big(-\frac{\hbar}{2}w_t\omega^1+\frac{\hbar}{2}v_t\omega^2+\frac{\hbar}{2}w\omega^1_t-\frac{\hbar}{2}v\omega^2_t-\frac{\hbar^2}{2m}v_x\omega^1_x-\frac{\hbar^2}{2m}w_x\omega^2_x\Big)dt\wedge dx \,.
\eeq
Then
\beq
p_1\mathcal{R}(d\lambda)=\mathcal{R}(p_1d\lambda)=- \frac{\hbar}{2} (w\omega^1\wedge dx-v\omega^2\wedge dx+\frac{\hbar}{m}v_x\omega^1\wedge dt+\frac{\hbar}{m}w_x\omega^2\wedge dt)
\eeq
and finally
\beq
\rho_1=\lambda+\frac{\hbar}{2} (w\omega^1\wedge dx-v\omega^2\wedge dx+\frac{\hbar}{m}v_x\omega^1\wedge dt+\frac{\hbar}{m}w_x\omega^2\wedge dt ).
\eeq
Now the following (and last) step is $\rho_2=\rho_1-p_2\mathscr{R}(d\rho_1)$. We have:
\beq
p_2d\rho_1=\frac{\hbar}{2}\Big(\omega^2\wedge\omega^1\wedge dx-\omega^1 \wedge \omega^2\wedge dx+\frac{\hbar}{m}\omega^1_x\wedge \omega^1\wedge dt+\frac{\hbar}{m}\omega^2_x\wedge \omega^2\wedge dt\Big).
\eeq
But then
$p_2\mathscr{R}(d\rho_1) =\mathscr{R}(p_2d\rho_1)=0$, \ie $\rho_2=\rho_1$ is just the Poincar\'e-Cartan form. This is due to the specific form of the Lagrangian, indeed the only non zero second partial derivatives of $\mathscr{L} (j^1\pi)$ contributes to \eqref{r=1} within vanishing terms involving double wedge products of same $\omega$'s.
\eEx

\subsection{Lepage equivalents for second order theories}

Let us now consider a second order Lagrangian $\lam=\mathscr{L}(x^i,y^\sig ,y^\sig_j,y^\sig_{jk})ds$. 
As a first step we compute the Poincar\'e-Cartan form of the Lagrangian using the first of Olga Rossi's recurrence formulae: $\rho_1=\lam-p_1\cR(d\lam)=\tht_\lam$.
\beq
& & d\lam =p_\sig\ome^\sig\wed ds+p^j_\sig\ome^\sig_j\wed ds+p^{jk}_\sig\ome^\sig_{jk}\wed ds= \\
& & =d_jd_k(p^{jk}_\sig\ome^\sig\wed ds)+d_j[(p^j_\sig-2d_kp^{jk}_\sig)\ome^\sig\wed ds] +(d_kd_jp^{jk}_\sig-d_jp^j_\sig+p_\sig)\ome^\sig\wed ds.
\eeq
Using the formula for the residual operator in this particular case:
\beq
\cR=\sum_{|I|=0}^1(-1)^1d_I\chi^{Ij}\wed ds_j,
\eeq
we obtain:
\beq
p_1\cR(d\lam)=-[(p^j_\sig-2d_kp^{jk}_\sig+d_kp^{jk}_\sig)\ome^\sig\wed ds_j+p^{jk}_\sig\ome^\sig_k\wed ds_j].
\eeq
Then, as expected, the Poincar\'e-Cartan form of the Lagrangian $\lam$ is:
\beq
\tht_\lam = \mathscr{L}ds+(p^j_\sig-d_kp^{jk}_\sig)\ome^\sig\wed ds_j+p^{jk}_\sig\ome^\sig_k\wed ds_j.
\eeq
We now rename $f^j_\sig:=p^j_\sig-d_kp^{jk}_\sig$ and $f^{jk}_\sig:=p^{jk}_\sig$ and continue computing the second of Rossi's recurrence formulae: $\rho_2=\tht_\lam-p_2\mathscr{R}(d\tht_\lam)$. 

We are interested only in the $2$-contact component of $d\tht_\lam$ because $p_2\mathscr{R}=p_2\mathscr{R}p_2$, \ie the Residual operator $\mathscr{R}$ does not increase the order of contactness of its argument, thus let us write
\beq
& & p_2(d\tht_\lam) =  \ome^\sig\wed(\der_\sig f^i_{\sig_2}\ome^{\sig_2}\wed ds_i-\der_{\sig_1}^{j_1}f^i_\sig\ome^{\sig_1}_{j_1}\wed ds_i-\der_{\sig_1}^{j_1j_2}f^i_\sig\ome^{\sig_1}_{j_1j_2}\wed ds_i + \\
& & +\der_\sig f^{ij_2}_{\sig_2}\ome^{\sig_2}_{j_2}\wed ds_i)+\ome^\sig_j\wed(\der_\sig^jf^{ij_2}_{\sig_2}\ome^{\sig_2}_{j_2}\wed ds_i-\der_{\sig_1}^{j_1j_3}f^{ij}_\sig\ome^{\sig_1}_{j_1j_3}\wed ds_i) \,.
\eeq
Now we recast this expression in the form $\sum_{|I|=0}^1d_I(\ome^\sig\wed\xi^I_\sig)$ and apply the integration by parts lemma, so to get for the Residual operator
\beq
& & p_2\mathscr{R}(d\tht_\lam) = \frac{1}{2}\der_\sig^jf^{ij_2}_{\sig_2}\ome^{\sig_2}_{j_2}\wed ds_{ij}-\frac{1}{2}\der_{\sig_1}^{j_1j_3}f^{ij}_\sig\ome^{\sig_1}_{j_1j_3}\wed ds_{ij}= \\
& & =\frac{1}{2}\der_\sig^jf^{ij_2}_{\sig_2}\ome^{\sig_2}_{j_2}\wed ds_{ij} \,,
\eeq 
being $f^{ij}_\sig$ symmetric in $(ij)$. Hence we obtain the following formula for $\rho_2$:
\beq
\rho_2=\mathscr{L}ds+f^i_\sig\ome^\sig\wed ds_i+f^{ij}_\sig\ome^\sig_j\wed ds_i+\frac{1}{2}\der^{i_1}_{\sig_1} f^{i_2j_2}_{\sig_2}\ome^{\sig_1}\wed\ome^{\sig_2}_{j_2}\wed ds_{i_1i_2}
\eeq
It turns out that it is possible to apply this reasoning at any step of Rossi's recurrence formulae, \eg at the third step we get 
\beq
p_3\mathscr{R}(d\rho_2)=-\frac{1}{6}\der^{i_1}_{\sig_1}\der^{i_2}_{\sig_2}f^{i_3j}_\sig\ome^{\sig_1}\wed\ome^{\sig_2}\wed\ome^\sig_j\wed ds_{i_1i_2i_3}.
\eeq
Proceeding in this way, the final expression of $\rho_n$ has the form:
\beq
& & \rho_n =\mathscr{L}ds+f^i_\sig\ome^\sig\wed ds_i+ \\
& &  +\sum_{q=1}^n\frac{1}{q!}\frac{\der\mathscr{L}}{\der y^{\sig_1}_{i_1}\dots\der y^{\sig_{q-1}}_{i_{q-1}}\der y^{\sig_q}_{i_qj}}\,\ome^{\sig_1}\wed\dots\wed\ome^{\sig_{q-1}}\wed\ome^{\sig_q}_j\wed ds_{i_1\dots i_q}
\eeq
We remark that the integration by parts is not uniquely defined in the case $r=2$. Indeed,  if we take into account in the decomposition of $p_2(d\rho_1)=p_2(d \tht_\lam )$ also the term
$\ome^\sig_j\wed(\der^j_\sig f^i_{\sig_2}\ome^{\sig_2}\wed ds_i)$ and work on it as above,
we get an additional term in $\rho_n$, thus obtaining a sort of generalization at the second order of the Krupka-Betounes equivalent, namely :
\beq
 \rho_n =\mathscr{L}ds+\sum_{q=1}^n\frac{1}{q!}\frac{\der\mathscr{L}}{\der y^{\sig_1}_{i_1}\dots\der y^{\sig_{q-1}}_{i_{q-1}}\der y^{\sig_q}_{i_qj}}\,\ome^{\sig_1}\wed\dots\wed\ome^{\sig_{q-1}}\wed\ome^{\sig_q}_j\wed ds_{i_1\dots i_q}+ \\
 + f^i_\sig\ome^\sig\wed ds_i +\sum_{q=1}^{n-1}\frac{1}{(q+1)!}\frac{\der f^{i_{q+1}}_{\sig_{q+1}}}{\der y^{\sig_1}_{i_1}\dots\der y^{\sig_q}_{i_q}}\,\ome^{\sig_1}\wed\dots\wed\ome^{\sig_q}\wed\ome^{\sig_{q+1}}\wed ds_{i_1\dots i_qi_{q+1}} \,,
\eeq
which exactly reduces to the Krupka-Betounes equivalent  \eqref{r=1}, when the Lagrangian is of order $r=1$ (see \cite{Krupkova09} for a review on Lepage equivalents of order $r\geq 1$). 

The study of the properties of this Lepage equivalent according to \cite{Betounes84} (see also \cite{VoGaVa21}) will be the subject of a separate paper.

\bRm 
We point out that Rossi's recurrence formulae provide a well defined second order Lepage equivalent $\rho_n$,  satisfying the well known defining  properties:
\begin{itemize}
\item $h\rho_n=\lambda$;
\item $\mathcal{E}_\lambda=p_1d\rho_{n}$ is a source form.
\end{itemize}
Indeed, the first condition is trivially verified since all the terms that are added to the Lagrangian in the recurrence steps are contact and thence do not affect the horizontal part of $\rho_n$. 

Furthermore, we note that the only terms in $\rho_n$ that can contribute to the $1$-contact component of the differential are the horizontal terms (they contribute via the contact differential $d_C$) and the $1$-contact terms (they contribute via the horizontal differential $d_H$).
In other words, we have $p_1d\rho_n=p_1d\rho_1$. 

Now,  $\rho_1$ is a distinguished Lepage equivalent, the so-called {\em principal Lepage equivalent};
hence $p_1d\rho_1$ is by definition a source form and this shows that $\rho_n$ is indeed a Lepage equivalent too. The proof for any order $r$ can be done in analogy with  \cite{Palese:VariationalSeq}, Theorem $3.11$. 
\eRm

\section*{Acknowledgements}

The first author (MP) was supported by 
the Department of Mathematics - University of Torino project $PALM$\_$RILO$\_$20$\_$01$ 
and would like to acknowledge the contribution of the COST Action CA17139. 
The third author (FZ) was also supported by a PhD grant of the University of G\"ottingen. 



\end{document}